\documentclass[twocolumn]{autart}    

\usepackage{graphicx}        
\usepackage{epstopdf} 
\usepackage{amsmath,amssymb,amsfonts}
\usepackage{subfig}
\usepackage{tikz}
\usepackage{pgfplots}
\pgfplotsset{compat=1.18}
\usetikzlibrary{arrows.meta,calc,decorations.markings,3d}
\usetikzlibrary{positioning}
\newtheorem{theorem}{Theorem}[section]
\newtheorem{lemma}[theorem]{Lemma}
\newtheorem{proposition}[theorem]{Proposition}

\theoremstyle{definition}
\newtheorem{definition}[theorem]{Definition}
\newtheorem{remark}[theorem]{Remark}

\providecommand{\qedsymbol}{$\square$}

\newenvironment{proof}[1][Proof]{%
  \par\noindent{\itshape #1.\ }\ignorespaces
}{%
  \unskip\nobreak\hfill\qedsymbol\par
}

\newcommand{\R}{\mathbb{R}}
\newcommand{\CP}{\mathbb{CP}}

\newcommand{\grad}{\operatorname{grad}}
\newcommand{\frakg}{\mathfrak{g}}

\newcommand{\ket}[1]{\lvert #1 \rangle}
\newcommand{\bra}[1]{\langle #1 \rvert}

\begin{document}

\begin{frontmatter}

\title{Quantum Riemannian Cubics with Obstacle Avoidance for Quantum Geometric Model Predictive Control \thanksref{footnoteinfo}} 

\thanks[footnoteinfo]{The author acknowledge financial support from Grant PID2022-137909NB-C21 funded by MCIN/AEI/ 10.13039/501100011033. The research leading to these results was supported in part by iRoboCity2030-CM, Robótica Inteligente para Ciudades Sostenibles (TEC-2024/TEC-62).}

\author[LC]{Leonardo J. Colombo}\ead{leonardo.colommbo@csic.es}    

\address[LC]{Centre for Automation and Robotics (CSIC-UPM), Ctra. M300 Campo Real, Km 0,200, Arganda del Rey, Madrid, Spain.} 

\begin{keyword}                           
Geometric model predictive control; Quantum control; Riemannian cubics; Variational integrators.               
\end{keyword}                             

\begin{abstract}
We propose a geometric model predictive control framework for quantum
systems subject to smoothness and state constraints. By formulating
quantum state evolution intrinsically on the projective Hilbert space,
we penalize covariant accelerations to generate smooth trajectories in
the form of Riemannian cubics, while incorporating state-dependent
constraints through potential functions. A structure-preserving
variational discretization enables receding-horizon implementation, and
a Lyapunov-type stability result is established for the closed-loop
system. The approach is illustrated on the Bloch sphere for a two-level
quantum system, providing a viable pathway toward predictive feedback
control of constrained quantum dynamics.
\end{abstract}

\end{frontmatter}
\section{Introduction}

From a geometric viewpoint, pure quantum states are naturally described as
elements of a reduced configuration space obtained by quotienting out the
physically irrelevant global phase. This reduction leads to the projective
Hilbert space $\mathbb{P}(\mathcal H)$, which can be identified with a complex
projective manifold and, in finite dimensions, with a homogeneous space of the
unitary group. This intrinsic formulation provides a coordinate-free
description of quantum state evolution that is consistent with the modern
geometric treatment of quantum control systems. Such a perspective has been
widely adopted in the control literature, where closed quantum systems are
modeled as bilinear control systems evolving on Lie groups or homogeneous
spaces, with controllability and reduction properties characterized in terms of
the associated Lie algebra structure
\cite{928587,altafini2002,altafiniTAC2012,boscain2006}.

Within this geometric setting, quantum control problems can be cast as
continuous-time optimal control problems on reduced state manifolds. The
resulting dynamics admit variational and Hamiltonian formulations analogous to
those encountered in classical nonlinear control, including Hamilton--Jacobi--Bellman
(HJB) characterizations of optimal feedback laws
\cite{nurdin2009,6314989}. While the exact solution of the quantum HJB equation
is generally intractable, recent work has shown that approximate
learning-based methods can be developed while preserving stability and
optimality guarantees \cite{dongpetersen2023,wadi2026}.

From both practical and geometric viewpoints, regularization of quantum control
trajectories is essential in order to mitigate experimental imperfections,
bandwidth limitations, and decoherence effects \cite{mirrahimi2005,nurdin2009}.
Penalizing covariant acceleration leads naturally to higher-order variational
problems posed directly on the reduced quantum state space, in which the
critical trajectories are Riemannian cubics
\cite{brody2012quantum,abrunheiro2018general,noakes1989cubic,camarinha2001geometry}.
These curves provide intrinsically smooth interpolants compatible with the
underlying Riemannian geometry.

Practical quantum control problems are also inherently constrained, as certain
regions of the quantum state space may correspond to noise-sensitive or
physically undesirable configurations. Rather than imposing hard state
constraints, such limitations are more naturally encoded through smooth
state-dependent penalties \cite{khatib1986real}. Within a
geometric framework, this leads to Riemannian cubics with obstacle avoidance
\cite{bloch2017variational}, yielding smooth and constraint-aware quantum
trajectories.

While open-loop solutions of these variational problems provide valuable
insight, robust operation in the presence of disturbances and uncertainty
requires feedback. This motivates the use of model predictive control (MPC), in
which finite-horizon optimal control problems are solved repeatedly in a
receding-horizon fashion. Extending MPC to quantum systems evolving on nonlinear
state manifolds, however, poses significant
conceptual and computational challenges.

In this work, we introduce a unified geometric framework for quantum trajectory
generation and control on the reduced state space of pure quantum states. The
dynamics are formulated intrinsically on the projective Hilbert space endowed
with the Fubini--Study metric, and higher-order smoothness is enforced through a
second-order variational principle. We develop a structure-preserving Lie group
variational discretization of the resulting dynamics and embed the proposed
framework into a geometric model predictive control scheme for quantum systems.
A Lyapunov-type stability result is established for the resulting closed-loop
dynamics, and the approach is illustrated on the Bloch sphere for a two-level
quantum system.

The remainder of the paper is organized as follows. Section~2 reviews the
geometric preliminaries required to formulate higher-order variational problems
on curved state spaces. Section~3 introduces Riemannian cubics with obstacle
avoidance in a general geometric setting. Section~4 specializes the formulation
to quantum systems on the Bloch sphere. Sections~5 and~6 develop the Lie group
formulation of the quantum dynamics and derive structure-preserving variational
integrators. Finally, Section~7 introduces the geometric model predictive
control framework, establishes closed-loop stability results, and demonstrates
the effectiveness of the approach in a receding-horizon setting.

\section{Riemannian manifolds}\label{Sec: background}

We briefly recall the minimal geometric notions required to formulate
second-order variational problems on curved configuration spaces. Let $(Q,g)$ be an $n$-dimensional Riemannian manifold, where $g$ is a
positive-definite symmetric covariant $2$-tensor field. We denote by
$\langle \cdot,\cdot\rangle_q := g_q(\cdot,\cdot)$ the associated inner
product on $T_qQ$. The
associated norm of a tangent vector $v_q\in T_qQ$ is denoted by
$\|v_q\|_g=\langle v_q,v_q\rangle_q^{1/2}$.

The Riemannian metric induces a unique torsion-free and metric-compatible
connection $\nabla$, known as the Levi--Civita connection. This
connection allows for the intrinsic differentiation of vector fields
and plays a central role in defining higher-order derivatives along
curves. Given two smooth vector fields $X$ and $Y$ on $Q$, the covariant
derivative $\nabla_X Y$ generalizes the directional derivative to curved
spaces. The curvature tensor associated with $\nabla$ is defined by $R(X,Y)Z := \nabla_X\nabla_Y Z - \nabla_Y\nabla_X Z - \nabla_{[X,Y]}Z$, and measures the noncommutativity of covariant derivatives.

Let $q:I\subset\mathbb R\to Q$ be a smooth curve with velocity
$\dot q(t)\in T_{q(t)}Q$. The Levi--Civita connection induces a covariant
derivative along the curve, denoted by $\frac{D}{dt}$, which assigns to
any vector field $W$ along $q$ a new vector field $\frac{DW}{dt}$ along
$q$. Higher-order covariant derivatives along $q$ are defined
inductively by $\displaystyle{\frac{D^k}{dt^k}W := \frac{D}{dt}\!\left(\frac{D^{k-1}}{dt^{k-1}}W\right)}$, for $k\in\mathbb N$. In particular, the covariant acceleration of the curve $q$ is given by
$\frac{D}{dt}\dot q$, and the vanishing of this quantity characterizes
geodesics.

If $Q$ is complete, the Hopf--Rinow theorem \cite{Boothby} guarantees that any two
points in $Q$ can be connected by a minimizing geodesic. The associated
Riemannian distance is defined by the length of such curves. Geodesics
also define the Riemannian exponential map
$\exp_q:T_qQ\to Q$, which maps an initial velocity to the endpoint of the
corresponding geodesic at unit time. In a neighborhood of $q$, the
distance can be expressed as
$d(q,y)=\|\exp_q^{-1}y\|_g$.

All notions introduced above are metric-dependent but otherwise
completely general. In the remainder of the paper, we specialize to
Riemannian manifolds arising in quantum mechanics, most notably complex
projective spaces endowed with the Fubini--Study metric \cite{provost1980}. In the qubit
case, this metric is equivalent, up to a constant factor, to the standard
round metric on the two-sphere. Since constant rescalings of the metric
do not affect the Levi--Civita connection or the form of the geodesic
equations up to reparametrization, we will freely use either
representation when convenient.

\section{Riemannian Cubics with Obstacle Avoidance}

Let $(Q,g)$ be an $n$-dimensional Riemannian manifold and denote by $\nabla$ the corresponding Levi--Civita connection and by $R$ the Riemannian curvature tensor. Let $\gamma : [0,T] \to Q$ be a smooth curve, we denote by $D_t$ the covariant derivative along $\gamma$, as introduced
in Section~\ref{Sec: background}. The covariant acceleration of $\gamma$ is defined as $D_t \dot\gamma(t) \in T_{\gamma(t)}Q$. Riemannian cubics \cite{noakes1989cubic}, \cite{camarinha2001geometry}, are critical points of
the functional
\begin{equation}
\label{eq:RCfunctional}
  J_0(\gamma) = \int_0^T \frac{1}{2}\,\|D_t \dot\gamma(t)\|_g^2 \, dt
\end{equation}
subject to fixed boundary data
$(\gamma(0), \dot\gamma(0))$ and $(\gamma(T), \dot\gamma(T))$. To model obstacle avoidance, one introduces a smooth potential
$V : Q \to \R$ that blows up or becomes large when $\gamma(t)$ approaches
forbidden regions  (see \cite{bloch2017variational} for details). This leads to:

\begin{definition}
Given $(Q,g)$ and a smooth obstacle avoidance potential $V : Q \to \R$, a
\emph{Riemannian cubic with obstacle avoidance} is a curve
$\gamma : [0,T] \to Q$ that is a critical point of
\begin{equation}
\label{eq:RCfunctionalObstacle}
  J(\gamma) = \int_0^T \left[
    \frac{1}{2} \|D_t\dot\gamma(t)\|_g^2 + V(\gamma(t))
  \right] dt
\end{equation}
under variations that fix both $\gamma(0),\dot\gamma(0)$ and
$\gamma(T),\dot\gamma(T)$.
\end{definition}

We recall the intrinsic equations from \cite{bloch2017variational} satisfied by critical curves of
\eqref{eq:RCfunctionalObstacle}.

\begin{theorem}
\label{thm:ELgeneral}
Let $(Q,g)$ be a Riemannian manifold with Levi--Civita connection $\nabla$,
and let $V : Q \to \R$ be smooth. A smooth curve $\gamma : [0,T] \to Q$ with
fixed boundary values $(\gamma(0),\dot\gamma(0))$ and
$(\gamma(T),\dot\gamma(T))$ is a critical point of
\eqref{eq:RCfunctionalObstacle} if and only if it satisfies
\begin{equation}
\label{eq:ELgeneral}
  D_t^3 \dot\gamma
  + R\!\left(D_t\dot\gamma, \dot\gamma\right)\dot\gamma
  + \operatorname{grad} V(\gamma) = 0,
\end{equation}
where $D_t^3 \dot\gamma := D_t ( D_t ( D_t \dot\gamma))$ and
$\operatorname{grad} V$ is the Riemannian gradient of $V$, defined
implicitly by the relation $g_q\!\left(\operatorname{grad} V(q),\, v_q\right)=
  \mathrm{d}V(q)[v_q]$, $\forall\, v_q \in T_q Q$, where $\mathrm{d}V(q):T_qQ\to\mathbb{R}$ is the differential of $V$ at $q$.

\end{theorem}

\begin{remark}
When $V\equiv 0$, \eqref{eq:ELgeneral} reduces to the familiar equation for
Riemannian cubics \cite{noakes1989cubic}, \cite{camarinha2001geometry}:
\[
  D_t^3 \dot\gamma + R\!\left(D_t\dot\gamma, \dot\gamma\right)\dot\gamma = 0.
\]
The potential adds a ``force'' term $\operatorname{grad} V(\gamma)$ that
pushes the curve away from the obstacle region.
\end{remark}

The formulation above provides a geometric regularization mechanism for
trajectory generation on nonlinear manifolds. In the context of control,
penalizing covariant acceleration corresponds to limiting rapid
variations of the underlying control action, while the potential term
encodes state-dependent constraints in a smooth manner. These features
will play a central role in the development of structure-preserving
discretizations and predictive feedback control schemes in the
following sections.

\section{Quantum Riemannian Cubics with Obstacle Avoidance on the Bloch Sphere}

Next, we study variational trajectory generation problems for the evolution of pure quantum states of a two-level system (qubit), in the presence of forbidden or undesirable regions in state space. Our goal is to construct smooth state
trajectories that minimise covariant acceleration, while avoiding prescribed
obstacles, in a manner intrinsic to the geometry of quantum mechanics \cite{provost1980}, \cite{anandan1990}.

Let $\mathbb{C}^2$ denote the two-dimensional complex Hilbert space endowed with
the standard Hermitian inner product. The space of pure qubit states is defined
as the associated projective Hilbert space
$\mathbb{P}(\mathbb{C}^2)
  :=
  (\mathbb{C}^2 \setminus \{0\}) / \mathbb{C}^\ast$, where $\mathbb{C}^\ast := \mathbb{C}\setminus\{0\}$, so that physically
irrelevant global phase factors are identified. This space is
isomorphic to the complex projective line $\CP^1$, which admits a 
geometrically intuitive realisation as the two-dimensional sphere $S^2$,
commonly known as the Bloch sphere. Throughout this section, we therefore
identify $S^2 \subset \mathbb{R}^3$ with the space of pure qubit states and
endow it with the standard \emph{round metric}, namely the Riemannian metric
induced by the Euclidean inner product of $\mathbb{R}^3$.

An equivalent representation of pure states is provided by the density-matrix
formalism. Any pure qubit state can be represented by a rank-one orthogonal
projector of the form $\rho=\ket{\psi}\bra{\psi}$, which satisfies
$\rho^2=\rho$ and $\operatorname{Tr}(\rho)=1$. In this representation, pure
states admit the Bloch decomposition $\rho = \frac{1}{2}\left(\mathbf{1} + \vec n\cdot \vec\sigma\right)$, $\vec n\in \mathbb{R}^3,\ \|\vec n\|=1$, where $\mathbf{1}$ denotes the $2\times2$ identity matrix and
$\vec\sigma=(\sigma_x,\sigma_y,\sigma_z)$ are the Pauli matrices. The unit vector
$\vec n$, known as the \emph{Bloch vector}, provides a one-to-one identification
between $\CP^1$ and the sphere $S^2$. Under this identification, a pure qubit state admits the parametrisation $\ket{\psi}
  =
  \cos(\theta/2)\ket{0}
  + e^{i\phi}\sin(\theta/2)\ket{1}$, where $(\theta,\phi)$ are spherical coordinates associated with the
corresponding Bloch vector $\vec n\in S^2$. The interpretation of this
representation is illustrated in Figure~\ref{fig:bloch_sphere_basic}.


\begin{figure}[h!]
  \centering
  \begin{tikzpicture}[scale=1.8]
    \shade[ball color=white,opacity=0.8] (0,0) circle (1);

    \draw[->] (-1.3,0) -- (1.3,0) node[right] {$x$};
    \draw[->] (0,-1.3) -- (0,1.3) node[above] {$z$};
    \draw[->,dashed] (-0.9,-0.9) -- (0.9,0.9) node[right] {$y$};

    \fill (0,1) circle (0.03) node[above left] {$\ket{0}$};
    \fill (0,-1) circle (0.03) node[below left] {$\ket{1}$};

    \coordinate (O) at (0,0);
    \coordinate (R) at (0.6,0.8);
    \draw[thick,->,blue] (O) -- (R) node[above right] {$\vec n$};

    \draw[thin] (0.55,0.0) arc[start angle=0,end angle=53,radius=0.55];
    \node at (0.7,0.2) {$\theta$};
  \end{tikzpicture}
  \caption{Bloch representation of a qubit state
           $\rho = \frac{1}{2}\left(\mathbf{1} + \vec n\cdot \vec\sigma\right)$.
           Pure states lie on the surface ($\|\vec n\|=1$), while mixed states correspond to points
in the interior ($\|\vec n\|<1$).}
  \label{fig:bloch_sphere_basic}
\end{figure}
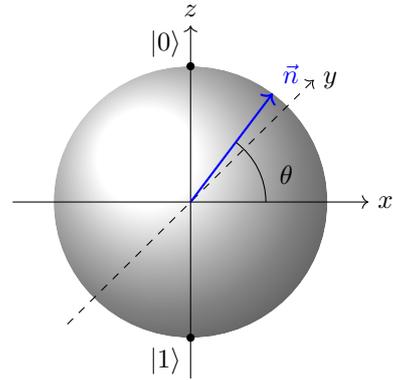

Let $\vec n(t)\in S^2$ be a smooth curve. Since $\|\vec n(t)\|=1$, differentiation yields $\vec n\cdot\dot{\vec n}=0$,
and hence $\dot{\vec n}(t)\in T_{\vec n(t)}S^2$.The Levi--Civita connection on $S^2$ associated with the round metric can be
expressed extrinsically by projecting the Euclidean derivative onto the tangent
space via the orthogonal projection $P(\vec n)$. Let $P(\vec n) = I - \vec n\,\vec n^\top$ denote the orthogonal projector onto
$T_{\vec n}S^2$ where $I$ is the
identity on $\mathbb{R}^3$. Thus $P(\vec n)$ acts on vectors in $\mathbb{R}^3$ and satisfies
$P(\vec n)\vec n = 0$ and $P(\vec n)^2 = P(\vec n)$. Next, we compute the covariant acceleration on $S^2$:

\begin{lemma}
\label{lem:Dtndot}
Let $\vec n : [0,T]\to S^2$ be a smooth curve. The covariant derivative
$D_t\dot{\vec n}$ along $\vec n$ is given by
\begin{equation}
\label{eq:Dtndot}
  D_t\dot{\vec n}
  = \ddot{\vec n}
  + (\dot{\vec n}\cdot\dot{\vec n})\,\vec n.
\end{equation}
\end{lemma}

\begin{proof}
Since $\|\vec n\|=1$, differentiating twice we obtain
$\vec n\cdot\dot{\vec n}=0$ and $\vec n\cdot\ddot{\vec n} =
- \|\dot{\vec n}\|^2$. The covariant derivative of the tangent vector
$\dot{\vec n}$ is the tangential component of $\ddot{\vec n}$:
$D_t\dot{\vec n} = P(\vec n)\ddot{\vec n}
  = \ddot{\vec n} - (\ddot{\vec n}\cdot \vec n)\,\vec n
  = \ddot{\vec n} + (\dot{\vec n}\cdot\dot{\vec n})\,\vec n$.
\end{proof}

For any $\vec n\in S^2$ and $a,b,c\in T_{\vec n}S^2$, the Riemann curvature
tensor satisfies (see \cite{Boothby} for instance) $R(a,b)c = (b\cdot c)\,a - (a\cdot c)\,b$.

For a smooth obstacle avoidance potential $V:S^2\to\R$, we consider the action functional
\begin{equation}
\label{eq:actionS2}
  J[\vec n] = \int_0^T \left[
    \frac{1}{2}\|D_t\dot{\vec n}(t)\|^2 + V(\vec n(t))
  \right] dt,
\end{equation}

\begin{definition}
A curve $\vec n : [0,T]\to S^2$ is called a
\emph{quantum Riemannian cubic with obstacle avoidance} if it is a critical
point of \eqref{eq:actionS2} under variations that fix both
$(\vec n(0),\dot{\vec n}(0))$ and $(\vec n(T),\dot{\vec n}(T))$.
\end{definition}

For boundary conditions we prescribe $\vec n(0) = \vec n_0$, $\dot{\vec n}(0) = \vec v_0$, $\vec n(T) = \vec n_T$, $\dot{\vec n}(T) = \vec v_T$, with $\vec n_0,\vec n_T\in S^2$ and $\vec v_0,\vec v_T$ tangent to the
sphere at the corresponding points. By Theorem \ref{thm:ELgeneral} and Lemma \ref{lem:Dtndot}, we can write
the quantum Riemannian cubics with obstacle avoidance on the Bloch sphere explicitly.

\begin{proposition}
\label{prop:explicitS2}
A smooth curve $\vec n : [0,T]\to S^2$ is a quantum Riemannian cubic with obstacle avoidance if 
\begin{equation}
\label{eq:explicitODE}
  P(\vec n)\Big(
   \ddddot{\vec n}
   + 2\|\dot{\vec n}\|^2 \ddot{\vec n}
   + 4(\dot{\vec n}\cdot\ddot{\vec n})\dot{\vec n}
   + \|\ddot{\vec n}\|^2 \vec n
   + \nabla V(\vec n)
  \Big) = 0,
\end{equation}
where $\nabla V$ is the Euclidean gradient of an extension of $V$ to a neighbourhood
of $S^2$ in $\R^3$.
\end{proposition}


\begin{proof}
Let $\vec n:[0,T]\to S^2\subset\R^3$ be a smooth curve and denote
$\vec v=\dot{\vec n}$, $\vec a=\ddot{\vec n}$. Since $\|\vec n\|=1$, we have
\begin{equation}\label{eq:basicS2ids}
\vec n\cdot \vec v=0,\qquad
\vec n\cdot \vec a=-\|\vec v\|^2,\qquad
\frac{d}{dt}(\vec n\cdot \vec a)=\vec v\cdot \vec a.
\end{equation}
With the round metric on $S^2$, the Levi--Civita covariant derivative along
$\vec n(t)$ is the tangential projection of the Euclidean derivative, hence
\begin{equation}\label{eq:Dtndot_recall}
D_t\vec v=P(\vec n)\vec a=\vec a+(\|\vec v\|^2)\vec n.
\end{equation}

We now compute the terms in \eqref{eq:ELgeneral} for $M=S^2$. First,
differentiating \eqref{eq:Dtndot_recall} in $\R^3$ gives $\displaystyle{\frac{d}{dt}(D_t\vec v)
=\dot{\vec a}+2(\vec v\cdot \vec a)\vec n+\|\vec v\|^2\vec v}$, and applying the tangential projection again yields
\begin{equation}\label{eq:Dt2ndot}
D_t^2\vec v
=P(\vec n)\Big(\dot{\vec a}+2(\vec v\cdot \vec a)\vec n+\|\vec v\|^2\vec v\Big).
\end{equation}
Differentiating once more and projecting, one obtains after grouping terms
(and using \eqref{eq:basicS2ids} repeatedly) 
\begin{equation}\label{eq:Dt3ndot}
D_t^3\vec v
=P(\vec n)\Big(
\ddddot{\vec n}
+2\|\vec v\|^2\vec a
+4(\vec v\cdot \vec a)\vec v
+\|\vec a\|^2\vec n
\Big).
\end{equation}


Using the curvature tensor formula of the unit sphere, taking
$a=D_t\vec v$, $b=\vec v$, $c=\vec v$ gives $R(D_t\vec v,\vec v)\vec v
=(\vec v\cdot\vec v)D_t\vec v-(D_t\vec v\cdot\vec v)\vec v$. Since $D_t\vec v=\vec a+\|\vec v\|^2\vec n$ and $\vec n\cdot\vec v=0$, 
$D_t\vec v\cdot\vec v=\vec a\cdot\vec v$, hence
\begin{equation}\label{eq:curvterm}
R(D_t\vec v,\vec v)\vec v
=\|\vec v\|^2(\vec a+\|\vec v\|^2\vec n)-(\vec v\cdot\vec a)\vec v.
\end{equation}

Finally, the intrinsic gradient of $V$ on $S^2$ is the tangential projection of
the Euclidean gradient of any smooth extension of $V$ to a neighbourhood of
$S^2$, i.e.
\begin{equation}\label{eq:gradproj}
\operatorname{grad}V(\vec n)=P(\vec n)\,\nabla V(\vec n).
\end{equation}
Substituting \eqref{eq:Dt3ndot}, \eqref{eq:curvterm}, and \eqref{eq:gradproj}
into \eqref{eq:ELgeneral}, and collecting terms inside $P(\vec n)$, yields
\eqref{eq:explicitODE}. The projection ensures the equation is intrinsic, i.e.
takes values in $T_{\vec n}S^2$.
\end{proof}

\begin{remark}
The constraint $\|\vec n(t)\|=1$ is preserved by \eqref{eq:explicitODE},
provided it holds at $t=0$ with $\vec n(0)\cdot\dot{\vec n}(0)=0$.
Indeed, differentiating $\|\vec n\|^2$ and using the fact that
the right-hand side of \eqref{eq:explicitODE} is tangent yields a
consistent evolution on $S^2$.
\end{remark}

 Given a smooth curve $\vec n:[0,T]\to S^2$, we introduce the augmented state
$(\vec n,\vec v,\vec a,\vec j)$ defined by $\vec v := \dot{\vec n}$,
  $\vec a := \ddot{\vec n}$ and $\vec j := \dddot{\vec n}$. The constraints $\|\vec n\|=1$ and $\vec n\cdot\vec v=0$ imply that, once the
equation of motion is imposed, the variables $\vec v$, $\vec a$ and $\vec j$
remain tangent to the sphere, i.e.\ $\vec v,\vec a,\vec j\in T_{\vec n}S^2$. Equation~\eqref{eq:explicitODE} can be rewritten on the
augmented state space as
\begin{align}
  \dot{\vec n} &= \vec v,\quad 
  \dot{\vec v} = \vec a,\quad
  \dot{\vec a} = \vec j,\quad
  \dot{\vec j} = F(\vec n,\vec v,\vec a,\vec j),
  \label{eq:sys4}
\end{align}
where the function $F$ is obtained by solving \eqref{eq:explicitODE} for the
fourth derivative $\ddddot{\vec n}$ and expressing all remaining terms in
terms of $(\vec n,\vec v,\vec a,\vec j)$ and $\nabla V(\vec n)$. Since $V$ is smooth on $S^2$, the right-hand side of \eqref{eq:sys4} is smooth.

The projection $P(\vec n)$ in \eqref{eq:explicitODE} ensures that the dynamics
is intrinsic to the sphere. In particular, if $\|\vec n(0)\|=1$ and
$\vec n(0)\cdot\vec v(0)=0$, then these constraints are preserved along the
flow: $\vec n(t)\in S^2$ and $\vec n(t)\cdot\vec v(t)=0$ for all times for which
the solution exists. This follows by differentiating $\|\vec n\|^2$ and
$\vec n\cdot\vec v$ and using the tangential nature of the right-hand side. Since the resulting first-order system \eqref{eq:sys4} defines
a smooth vector field, standard ODE theory guarantees local existence and
uniqueness of solutions for compatible initial data.

\begin{remark} The boundary value problem with fixed data
$(\vec n(0),\vec v(0))$ and $(\vec n(T),\vec v(T))$ is variational in nature and
is more delicate than the corresponding initial value problem. While local
existence and uniqueness of solutions to the system \eqref{eq:sys4} follow from ODE theory, the existence
of minimisers of the action functional \eqref{eq:actionS2} relies on compactness
and lower semicontinuity arguments. Such existence results for Riemannian
cubics with obstacle potentials are available in the literature; see, \cite{goodman15local}. 
\end{remark}

\subsection{Avoidance of undesired quantum states}\label{sec:obstacle}

We model forbidden configurations in the qubit state space by designing a smooth artificial potential that penalises proximity to
undesired quantum states. In the Bloch ball, a point on $S^2$
corresponds to a pure quantum state, and avoiding a ``point obstacle'' should
therefore be interpreted as suppressing the occupation of a specific
(unwanted) pure state or, more generally, of a leakage or escape channel in
state space.

We fix a reference undesired pure state, represented in the Bloch sphere by an
arbitrary unit vector $\vec n_\star\in S^2$. Avoidance of $\vec n_\star$
amounts to steering the quantum trajectory away from this state while
maintaining smoothness of the evolution. Since the Bloch-sphere representation
is invariant under unitary transformations, any undesired pure state can be
mapped to the north pole by a suitable rotation in $SU(2)$. Therefore, without
loss of generality, we adopt the convenient representative $\vec n_\star=(0,0,1)$. In the standard Bloch-sphere convention, this choice corresponds to the
computational basis state $\ket{0}$, so that avoidance of $\vec n_\star$
amounts to penalising population of $\ket{0}$.

Let $\theta(\vec n)$ denote the geodesic distance between the current quantum
state $\vec n\in S^2$ and the undesired state $\vec n_\star$ with respect to
the round metric, namely $\theta(\vec n)
  = d_{S^2}(\vec n,\vec n_\star)
  = \arccos(\vec n\cdot \vec n_\star)\in[0,\pi]$. We introduce an obstacle avoidance potential of the form
\begin{equation}\label{eq:potential}
  V_{\tau,D,N}(\vec n)
  := \frac{\tau}{1+\left(\frac{\theta(\vec n)}{D}\right)^{2N}},
  \end{equation} where $\tau>0$ sets the strength of the energetic penalty, $D>0$ defines a
tolerance radius around the undesired state, and $N\in\mathbb N$ controls the sharpness of
the barrier.

This potential assigns a maximal cost $V_{\tau,D,N}(\vec n_\star)=\tau$ to the
undesired quantum state and decays rapidly away from it, effectively creating
a repulsive barrier around the corresponding Bloch direction. For sufficiently
large $N$, the potential approximates a hard exclusion of a neighbourhood of
$\vec n_\star$, while remaining smooth and compatible with variational
formulations. Such potentials can be used in quantum control to suppress
population of leakage states or escape channels without introducing
discontinuous control actions.

The parameter $D$ defines the geodesic radius (or sensing radius) of the forbidden
region around the undesired state, corresponding to a bound on the admissible
population of that state. The sharpness parameter $N$ controls how abruptly the
penalty increases when approaching this region: small values of $N$ yield a
soft repulsive barrier, while larger values approximate a hard exclusion of a
spherical cap, without sacrificing smoothness. Finally, the amplitude $\tau$
sets the relative weight of the obstacle penalty in the variational problem.

\begin{remark}
Let $|\psi_\star\rangle$ denote the undesired pure state associated with the
Bloch vector $\vec n_\star$, and let $\mathcal{P}_\star=|\psi_\star\rangle\langle\psi_\star|$
be the corresponding rank-one projector. For a pure state $|\psi\rangle$ with
Bloch vector $\vec n$, one has $\langle\psi|\mathcal{P}_\star|\psi\rangle
= \tfrac12\big(\mathbf{1}+\vec n\cdot\vec n_\star\big)
= \cos^2\!\big(\tfrac{\theta(\vec n)}{2}\big)$, where $\theta(\vec n)=d_{S^2}(\vec n,\vec n_\star)$.
Therefore, the obstacle potential $V_{\tau,D,N}(\vec n)$ can be interpreted as nonlinear penalty on the population of the undesired quantum state.
In this sense, avoidance of a point obstacle in the Bloch sphere is equivalent
to suppressing population of a leakage or escape state in quantum control.
\end{remark}

\begin{remark}
The family $V_{\tau,D,N}$ is an \emph{avoidance family} (in the sense of
artificial potential design for modified Riemannian cubics \cite{goodman2021variational}): for fixed $\tau$ and $D$,
as $N\to\infty$ we have $V_{\tau,D,N}(\vec n)\to\tau$ whenever $\theta(\vec n)<D$,
and $V_{\tau,D,N}(\vec n)\to 0$ whenever $\theta(\vec n)>D$, with
$V_{\tau,D,N}(\vec n)=\tau/2$ for all $\vec n$ such that $\theta(\vec n)=D$. This matches the qualitative requirements used
to obtain obstacle-avoidance guarantees for minimisers when the problem is feasible.
\end{remark}

To gain geometric and control-oriented intuition on the obstacle avoidance
mechanism, we visualise the artificial potential $V_{\tau,D,N}$ introduced in
\eqref{eq:potential}. Since this potential depends only on the geodesic
distance to the undesired state $\vec n_\star$, its level sets define
spherical caps in the Bloch sphere. Figure~\ref{fig:obstacle_potential} illustrates the geometry of the obstacle
potential and the effect of the sharpness parameter $N$.

\begin{figure}[h!]
\begin{center}
\includegraphics[height=6cm]{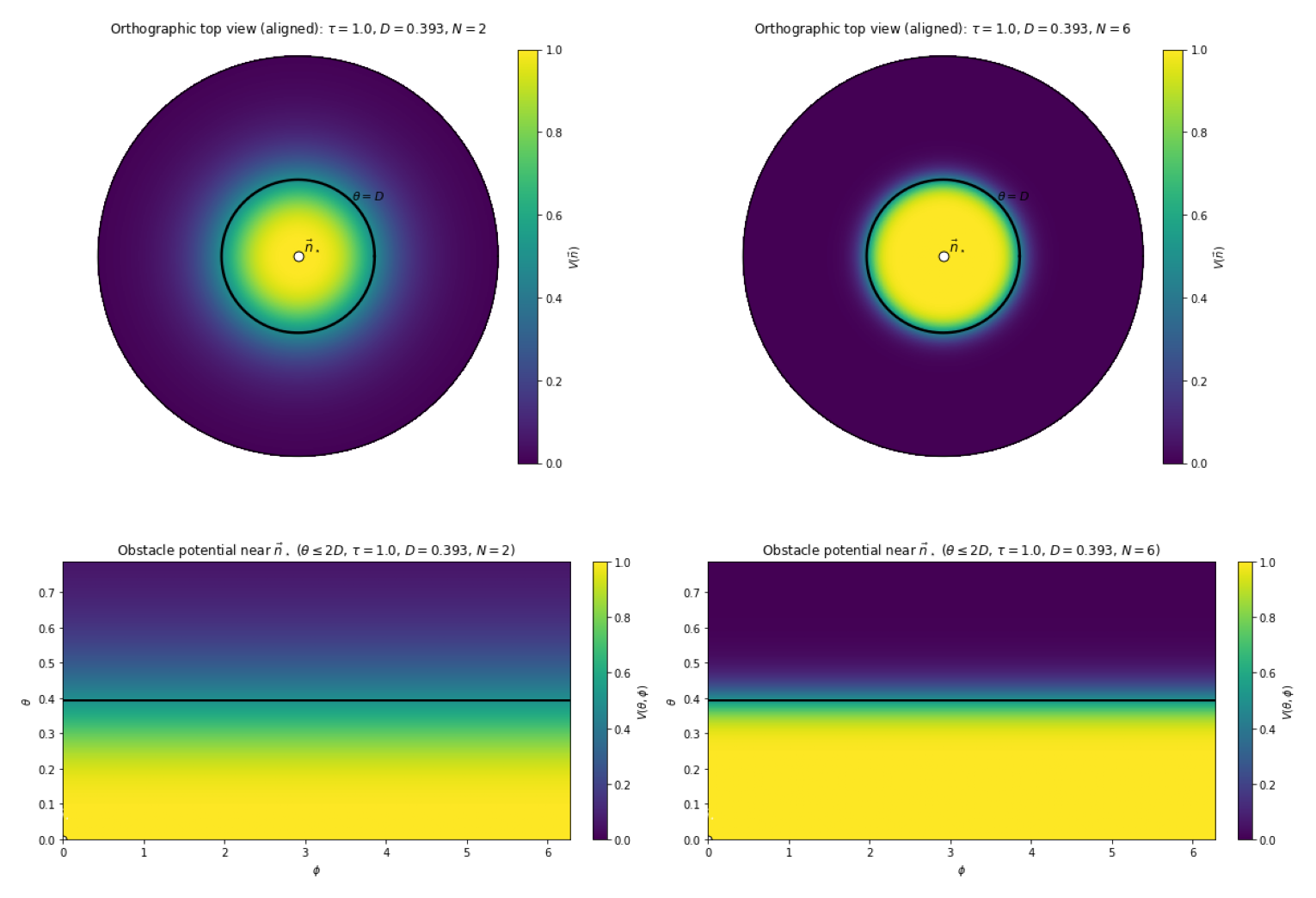}    
\caption{
Obstacle potential on the Bloch sphere.
\emph{Top row:} orthographic top views of the lifted obstacle potential
$V(\vec n)$, aligned with the obstacle direction $\vec n_o$, for fixed
$\tau = 1.0$ and $D = 0.393$, comparing the cases $N=2$ (left) and $N=6$
(right). The black contour indicates the level set $\theta = D$, defining
the effective obstacle boundary.
\emph{Bottom row:} meridional sections of the same potential in terms of
the polar angle $\theta$ for $\theta \leq 2D$, highlighting the increasing
steepness and localisation of the repulsive barrier as $N$ increases.
}
\label{fig:obstacle_potential}                                 
\end{center}                                 
\end{figure}

\medskip

\begin{lemma}\label{lem:gradV}
Assume $\vec n\cdot\vec n_\star \neq \pm1$ and let $u(\vec n):=\vec n\cdot \vec n_\star$.
Then
\begin{equation}\label{eq:gradtheta}
  \nabla \theta(\vec n)
  = -\frac{1}{\sqrt{1-u(\vec n)^2}}\,\vec n_\star,
\end{equation}
and
\begin{equation}\label{eq:dVdtheta}
  \frac{\partial V_{\tau,D,N}}{\partial \theta}
  = -\tau\,\frac{2N}{D^{2N}}\,
     \frac{\theta^{2N-1}}{\big(1+(\theta/D)^{2N}\big)^2}.
\end{equation}
Consequently, $\displaystyle{\nabla V_{\tau,D,N}(\vec n)
  = \frac{\partial V_{\tau,D,N}}{\partial \theta}\,\nabla\theta(\vec n)}$, where $\nabla$ denotes the Euclidean gradient of an extension to a neighbourhood
of $S^2\subset\R^3$. The Riemannian gradient is
\begin{equation}\label{eq:intrgradV}
  \grad_{S^2} V_{\tau,D,N}(\vec n) = P(\vec n)\,\nabla V_{\tau,D,N}(\vec n).
\end{equation}
\end{lemma}

\begin{proof}
From the formula of the geodesic distance, $\theta(\vec n)=\arccos(u(\vec n))$ and
$\nabla u(\vec n)=\vec n_\star$. Using $\frac{d}{du}\arccos(u)=-\frac{1}{\sqrt{1-u^2}}$
gives \eqref{eq:gradtheta}. Differentiating \eqref{eq:potential} w.r.t.\ $\theta$
yields \eqref{eq:dVdtheta}, and $\nabla V_{\tau,D,N}(\vec n)$ follows by the chain rule.
The intrinsic gradient on the embedded sphere is obtained by tangential
projection, giving \eqref{eq:intrgradV}.
\end{proof}

Substituting \eqref{eq:potential} and Lemma~\ref{lem:gradV} into the quantum Riemannian cubic with obstacle avoidance \eqref{eq:explicitODE} yields a fully explicit fourth-order ODE on $S^2$.

Whenever the boundary-value problem is feasible, sufficiently
large penalty parameters enforce avoidance of the undesired quantum state by
all minimisers:

\begin{proposition}\label{prop:avoidance}
Consider the $V_{\tau,D,N}$ defined in \eqref{eq:potential}. Fix $D>0$ and suppose the boundary-value problem is \emph{feasible} in the sense that
there exists at least one admissible curve joining the boundary data while remaining
outside the geodesic ball $\{\vec n\in S^2:\theta(\vec n)<D\}$ around $\vec n_\star$.
Then, for any $r\in(0,D)$ there exist $\tau_\star>0$ and $N_\star\in\mathbb N$ such that,
for all $\tau\ge \tau_\star$ and $N\ge N_\star$, every minimiser of the action \eqref{eq:actionS2} avoids the obstacle with tolerance $r$, i.e.\ $\theta(\vec n(t))\ge r$ for all $t\in[0,T]$.
\end{proposition}

\begin{proof}
The argument follows the standard obstacle-avoidance estimate for modified Riemannian
cubics \cite{goodman2021variational}. One compares the action of a minimiser against a feasible reference trajectory
that stays in the ``safety region'' (outside the $D$-ball), and uses lower/upper bounds
of $V_{\tau,D,N}$ on a ``risk region'' versus the safety region. For avoidance families,
increasing $\tau$ and $N$ makes $V_{\tau,D,N}$ arbitrarily large on $\{\theta<D\}$ while
making it arbitrarily small on $\{\theta\ge D\}$. This yields a contradiction if a
minimiser were to enter $\{\theta<r\}$: it would incur a potential cost that exceeds the
reference action bound. Hence minimisers remain outside $\{\theta<r\}$.
\end{proof}

\begin{remark}
If one requires $V$ to vanish \emph{identically} outside the sensing radius $D$, one may
replace \eqref{eq:potential} by a smooth bump-type potential supported on $\{\theta<D\}$,
e.g.\ $V(\vec n)=e^\tau\exp\!\big(-\frac{1}{1-(\theta(\vec n)/D)^{2N}}\big)$ for $\theta<D$ and
$V(\vec n)=0$ otherwise \cite{goodman2021variational}. 
\end{remark}

\subsection{Modeling general quantum obstacle}\label{sec:general_obstacle}

We consider as obstacle a (closed) forbidden region $P\subset S^2$, representing
a set of undesired pure states. Using the round metric on $S^2$, we denote the
geodesic distance by $d_{S^2}(\cdot,\cdot)$ and define the distance-to-set $d(\vec n,P) := \inf_{\vec m\in P} d_{S^2}(\vec n,\vec m)
  = \inf_{\vec m\in P} \arccos(\vec n\cdot \vec m)$. Following \cite{goodman2021variational} for a tolerance $D>0$ we
seek trajectories that remain in $S^2\setminus B_D(P)$, where
$B_D(P)=\{\vec n\in S^2: d(\vec n,P)<D\}$. 

Let $f_{\tau,D,N}:[0,\pi]\to\R$ be the smooth radial profile
\begin{equation*}
  f_{\tau,D,N}(s)=\frac{\tau}{1+(s/D)^{2N}},\,\,\, \tau>0,\ D>0,\ N\in\mathbb N.
\end{equation*}
Define the potential $V_{\tau,D,N}(\vec n) := f_{\tau,D,N}\big(d(\vec n,P)\big)$. For fixed $\tau$ and $D$, as $N\to\infty$ the profile $f_{\tau,D,N}$ approaches
$\tau$ on $[0,D)$ and to $0$ on $(D,\pi]$, so that the family acts as an
``avoidance family'' for the region $P$ at scale $D$.


To obtain an explicit, smooth potential depending only on point-to-point
distances, we assume that $P$ is totally bounded and choose a finite cover by
geodesic balls $P \subset \bigcup_{i=1}^M B_\varepsilon(\vec p_i)$, $\vec p_i\in S^2$, for some $\varepsilon>0$. We then treat $\{\vec p_i\}_{i=1}^M$ as point-obstacles
and define the potential
\begin{equation}\label{eq:sum_potential}
  V_{\tau,D,N}^{\mathrm{cov}}(\vec n)
  :=\sum_{i=1}^M \frac{\tau}{1+\left(\frac{\theta_i(\vec n)}{D}\right)^{2N}},
\end{equation} with $\theta_i(\vec n):=d_{S^2}(\vec n,\vec p_i)$. 
Each term is smooth away from $\vec n=\pm \vec p_i$, and its Euclidean gradient
is given explicitly by Lemma~\ref{lem:gradV} with $\vec n_\star$ replaced by
$\vec p_i$. The intrinsic gradient on $S^2$ is obtained by tangential projection,
$\grad_{S^2}V_{\tau,D,N}^{\mathrm{cov}}(\vec n)=P(\vec n)\nabla V_{\tau,D,N}^{\mathrm{cov}}(\vec n)$.


\begin{proposition}
\label{prop:avoidance_general}
Let $P\subset S^2$ be a closed forbidden region and fix a tolerance radius $D>0$.
Assume that the boundary-value problem for \eqref{eq:RCfunctionalObstacle} admits at least
one admissible curve satisfying the prescribed boundary data and remaining in
the safe set $S^2\setminus B_D(P)$, where $B_D(P)=\{\vec n\in S^2:\ d(\vec n,P)<D\}$. Then, for any $r\in(0,D)$ there exist $\tau_\star>0$ and $N_\star\in\mathbb N$ such
that, for all $\tau\ge \tau_\star$ and $N\ge N_\star$, every minimiser of \eqref{eq:RCfunctionalObstacle} with obstacle potential $V_{\tau,D,N}$ (or, in the practical case, with
the cover-based potential $V^{\mathrm{cov}}_{\tau,D,N}$) satisfies $d(\vec n(t),P)\ge r$ for all $t\in[0,T]$.\end{proposition}

\begin{proof}
The argument follows the standard sufficient-conditions framework for obstacle
avoidance in modified Riemannian cubics \cite{goodman2021variational}. Let $\gamma_{\mathrm{ref}}:[0,T]\to S^2$ be a feasible reference curve satisfying
the boundary data and remaining in the safety region
$S^2\setminus B_D(P)$. Since $V_{\tau,D,N}$ is an avoidance family, there exists
a constant $C>0$ such that $V_{\tau,D,N}(\vec n)\le C$
for all $\vec n\in S^2\setminus B_D(P)$, uniformly in $N$, while $V_{\tau,D,N}(\vec n)\to\tau$ as $N\to\infty$ for all $\vec n\in B_D(P)$.

Let $\vec n:[0,T]\to S^2$ be a minimiser of the action. Suppose, by contradiction,
that $\vec n(t_0)\in B_r(P)$ for some $t_0\in[0,T]$. By continuity, there exists
an interval $I\subset[0,T]$ of positive measure on which $\vec n(t)\in B_D(P)$.
On this interval, the potential term satisfies $\displaystyle{V_{\tau,D,N}(\vec n(t))\ge \frac{\tau}{2}}$ for $N$ sufficiently large. Consequently, the action of $\vec n$ admits the lower bound $\displaystyle{J[\vec n]\ge \frac{\tau}{2}\,|I|}$. On the other hand, the action of the reference curve $\gamma_{\mathrm{ref}}$
remains bounded independently of $\tau$ and $N$, since it stays entirely in the
safety region. Choosing $\tau$ sufficiently large yields $J[\vec n] > J[\gamma_{\mathrm{ref}}]$, which contradicts the minimality of $\vec n$. Therefore, for $\tau$ and $N$ large enough, no minimiser can enter $B_r(P)$.

For the potential $V^{\mathrm{cov}}_{\tau,D,N}$, the same argument
applies since $P$ is contained in a finite union of geodesic balls and $V^{\mathrm{cov}}_{\tau,D,N}$ preserves the avoidance-family properties.
\end{proof}

\section{Lie Group Formulation of Second-Order Variational Quantum Dynamics}
From a geometric viewpoint, pure quantum states are naturally described as
elements of a reduced configuration space obtained by quotienting out the
physically irrelevant global phase. This construction endows the space of pure
states with a canonical Riemannian structure, induced by quantum transition
probabilities, which provides an intrinsic notion of distance, velocity, and
covariant acceleration. This framework allows higher-order
variational problems for quantum state trajectories to be formulated directly
on the physical state space, without reference to a particular phase
representation.

\subsection{Riemannian homogeneous spaces}

Let $G$ be a Lie group and its Lie algebra $\mathfrak{g}$ be defined as the tangent space to $G$ at the identity, $\mathfrak{g}:=T_{e}G$. Denote by $L_g$ the left-translation map $L_g:G\to G$ given by $L_g(h)=gh,$ for all $g,h\in G,$ where $gh$ denotes the Lie group multiplication between the elements $g$ and $h$.

\begin{definition}\label{def: left-action}
Let $G$ be a Lie group and $H$ be a smooth manifold. A \textit{left-action} of $G$ on $H$ is a smooth map $\Phi: G \times H \to H$ such that for all $g,h \in G,$ $q \in H$ (i) $\Phi(e, q) = q$, (ii) $\Phi(g, \Phi(h, q)) = \Phi(gh, q)$, and (iii) the map $\Phi_g: H \to H$ defined by $\Phi_g(q) = \Phi(g, q)$ is a diffeomorphism.
\end{definition}

We now define an important class of group actions:

\begin{definition}\label{def: types of group actions}
Let $G$ be a Lie group, $H$ be a smooth manifold, and $\Phi: G \times H \to H$ be a left-action. We say that $\Phi$ is \textit{transitive} if for every $p,q \in H$, there exists some $g \in G$ such that $gp = q$. \end{definition}

One of the most important cases of a Lie group action is a Lie group acting on itself. The left translation $L_{g}$ is a left action of $G$ on itself and a diffeomorphism on $G$.  Its tangent map  (i.e, the linearization or tangent lift) is denoted by $T_{h}L_{g}:T_{h}G\to T_{gh}G$. Consider an inner-product on the Lie algebra $\mathfrak{g}$ denoted by $\langle \cdot,\cdot\rangle_\mathfrak{g}$. Using the left-translation we define a Riemannian metric on $G$ by the relation $\langle X,Y\rangle:=\langle (T_{g} L_{g^{-1}}) (X),(T_{g} L_{g^{-1}}) (Y)\rangle_\mathfrak{g}$ for all $g\in G, \; X,Y\in T_g G$, which is called left-invariant metric because $\langle (T_{g}L_{h})(X),(T_{g}L_{h})(X)\rangle=\langle X,Y\rangle$ for all $g, h\in G, \; X,Y\in T_g G$.

\begin{definition}Let $G$ be a connected Lie group equipped with a left-invariant Riemannian metric. A \textit{homogeneous space} $H$ of $G$ is a smooth manifold on which $G$ acts transitively. Any Lie group is itself a homogeneous space, where the transitive action is given by left-translation. \end{definition}

Suppose that $\Phi: G \times H \to H$ is a transitive left-action, which we denote by $gx := \Phi_g(x)$. It can be shown that for any $x \in H$, we have $G/\text{Stab}(x) \cong H$ as differentiable manifolds, where $\text{Stab}(x) := \{g \in G \ \vert \ gx = x\}$ denotes the \textit{stabilizer subgroup} (also called the \textit{isotropy subgroup}) of $x$, and $G/\text{Stab}(x)$ denotes the space of equivalence classes determined by the equivalence relation $g \sim h$ if and only if $g^{-1}h \in \text{Stab}(x)$. In addition, for any closed Lie subgroup $K \subset G$, the left-action $\Phi: G \times G/K \to G/K$ satisfying $\Phi_g([h]) = [gh]$ for all $g, h \in G$ is transitive, and so $G/K$ is a homogeneous space. Hence, we may assume without loss of generality that $H := G/K$ is a homogeneous space of $G$ for some closed Lie subgroup $K$. Let $\pi: G \to H$ be the canonical projection map. With this notation, the left multiplication commutes with the left action in the sense that $\Phi_{g}\circ \pi=\pi\circ L_{g}$ for any $g\in G$.

 Given a left-invariant
metric on $G$, any curve $g(t)\in G$ can be described via its left-trivialised
velocity $\xi(t) = g(t)^{-1} \dot g(t)\in\frakg$. The reconstruction equation is $\dot g(t) = g(t)\,\xi(t)$. If $M=G/K$ is a homogeneous space, a curve $\gamma(t)\in M$ can be lifted
to a curve $g(t)\in G$ such that $\gamma(t)=\pi(g(t))$.

The pure-state space $\CP^n$ is a compact homogeneous manifold obtained as a
quotient of the special unitary group by the stabiliser of a reference ray.
More precisely, $\CP^n \;\simeq\; G/K$ with $G := SU(n+1)$ and $K := S\big(U(1)\times U(n)\big)$. The group $G$ acts transitively on $\CP^n$ via the natural projective action
$[U]\cdot[\psi]:= [U\psi]$, where $[\psi]\in\CP^n$ denotes the equivalence class of $\psi\in\mathbb C^{n+1}$
under multiplication by a nonzero complex scalar. 

In the qubit case $n=1$, the isotropy subgroup reduces to
$K\simeq U(1)$, and one recovers the standard identifications $\CP^1 \;\simeq\; SU(2)/U(1) \;\simeq\; S^2$. The corresponding projection map $\pi:SU(2)\to S^2$ assigns to a unitary matrix
$U$ the Bloch vector associated with the pure state $U\ket{\psi_0}$, where
$\ket{\psi_0}$ is a fixed reference state.

\subsection{Second-order variational problems on Lie groups}\label{sec:LG-symmbreak}

We now recall the Lie-group framework for second-order variational problems and
the treatment of obstacle-type potentials via \emph{reduction with symmetry breaking}.
The guiding principle is that the kinetic term (built from a left-invariant metric)
is typically $G$-invariant, whereas the obstacle potential is not, and must therefore be
handled through an \emph{extended} potential that restores invariance in an augmented space;
see \cite{goodmancolombo_reduction_broken,goodmancolombo_EP_broken,holm2025geometric}.

Let $G$ be a Lie group endowed with a left-invariant Riemannian metric and Levi--Civita
connection $\nabla$. Consider a smooth curve $g(t)\in G$ and define the left-trivialised
velocity $\xi(t)=g(t)^{-1}\dot g(t)\in\mathfrak g$. For second-order variational problems, it is convenient to work with the left-trivialised
covariant acceleration. In the Riemannian cubic-with-obstacle setting we consider a second-order Lagrangian
of the form
\begin{equation}\label{eq:Lgroup_clean}
  L(g,\dot g,\ddot g) \;=\; \frac12 \|\eta\|^2 \;+\; V_{\mathrm{ext}}(g,\alpha),
\end{equation}
where $V_{\mathrm{ext}}$ is an \emph{extended potential} depending on an auxiliary parameter
$\alpha$ ranging over a smooth manifold $M$.

A potential $V:G\to\R$ is rarely left-invariant in applications (obstacles are attached to
specific configurations), so the symmetry of the kinetic term is \emph{broken}. Following
\cite{goodmancolombo_EP_broken}, we say that $V$ admits a \emph{partial symmetry} over $M$
if there exist a left action $\Psi:G\times M\to M$ and a smooth function
$V_{\mathrm{ext}}:G\times M\to\R$ such that:
(i) for some $\alpha_0\in M$, $V_{\mathrm{ext}}(g,\alpha_0)=V(g)$ for all $g\in G$; and
(ii) $V_{\mathrm{ext}}$ is invariant under the diagonal action
$\chi_h(g,\alpha)=(hg,\Psi_h(\alpha))$, i.e.
$V_{\mathrm{ext}}(hg,\Psi_h(\alpha)) \;=\; V_{\mathrm{ext}}(g,\alpha)$, $\forall\,g,h\in G,\ \alpha\in M$.
This construction is the standard mechanism enabling reduction with symmetry breaking \cite{holm2025geometric},
and it is particularly natural for obstacle avoidance, where $M$ encodes the obstacle
location (or a finite cover describing an obstacle region). \cite{goodmancolombo_reduction_broken,goodmancolombo_EP_broken}

When the configuration space is a homogeneous space $G/K$ (e.g. $\CP^1\simeq SU(2)/U(1)\simeq S^2$),
a curve $\gamma(t)\in G/K$ can be lifted to $g(t)\in G$ with $\gamma=\pi(g)$.
A potential $V:G/K\to\R$ induces a lifted potential $\widetilde V(g)=V(\pi(g))$ on $G$.
If $V$ is $G$-invariant it reduces directly; otherwise, one introduces an extended potential
$V_{\mathrm{ext}}$ invariant under $\Psi$ so that the variational principle
 can still be reduced in the presence of
broken symmetries. \cite{goodmancolombo_reduction_broken,goodmancolombo_EP_broken}

By construction, the potential $V_{\mathrm{ext}}:G\to\mathbb R$
is invariant under the \emph{right} action of the isotropy subgroup $K$, $V_{\mathrm{ext}}(gk) = V_{\mathrm{ext}}(g)$, $\forall\,k\in K$, which ensures that $V_{\mathrm{ext}}$ descends to a well-defined function
on the homogeneous space $G/K$. In contrast, $V_{\mathrm{ext}}$ is generally \emph{not} invariant under the
\emph{left} action of $G$, corresponding to a symmetry breaking of the
homogeneous-space dynamics.

\begin{remark}
For compact groups with bi-invariant metrics (notably $SU(2)$), the expressions for the
connection and curvature simplify, and the reduced second-order equations take a particularly form in terms of $\xi(t)$ (and the parameter dynamics associated with $\alpha$) \cite{goodmancolombo_reduction_broken,goodmancolombo_EP_broken}.
\end{remark}

\subsection{Axial symmetry and conserved quantities}


We now consider an important special case of physical relevance, namely obstacle
potentials depending only on the expectation value of the Pauli operator
$\sigma_z$.
Equivalently, we assume that the potential has the form $V(\vec n) = \widetilde V(n_z)$,\, $n_z := \vec n \cdot e_z$, that is, $V$ depends only on the $z$--component of the Bloch vector, which
coincides with the quantum expectation value
$\langle \sigma_z \rangle$ for pure states.
Such potentials arise naturally when the forbidden or undesirable region is
axially symmetric about the $z$--axis, as in population constraints or
decoherence--sensitive states.
In quantum control, this corresponds to penalising the expectation value of
$\sigma_z$, or equivalently the population of a preferred computational basis
state, a modelling choice commonly adopted in optimal and robust control
formulations \cite{bonnard2012quantum,dong2010quantum}.

\begin{remark}
In the Bloch--sphere representation, the expectation value
$\langle\sigma_z\rangle$ coincides with the $z$--component $n_z$ of the Bloch
vector and admits a direct interpretation in terms of measurement statistics.
For a normalised pure state $\ket{\psi}$,
$\langle\sigma_z\rangle=|\langle 0|\psi\rangle|^2 - |\langle 1|\psi\rangle|^2
  =
  2|\langle 0|\psi\rangle|^2 - 1$, where $|\langle 0|\psi\rangle|^2$ and $|\langle 1|\psi\rangle|^2$ are the
probabilities of obtaining the outcomes $+1$ and $-1$, respectively, in a
projective measurement of $\sigma_z$.
Thus, invariance under rotations generated by $\sigma_z$ corresponds to
preserving population differences between the computational basis states.
Potentials depending only on $n_z$ therefore model smooth penalties on the
population of a preferred (or forbidden) state.
\end{remark}

Rotations about the $z$--axis define a one--parameter subgroup
$SO(2)\subset SO(3)$ acting on $S^2$ by $\Phi_\theta(\vec n) := R_\theta \vec n,\,
   R_\theta\in SO(2)$. Since $\Phi_\theta$ is an isometry of the round metric, it induces
a natural action on the second--order tangent bundle $T^2S^2$, $(\vec n,\dot{\vec n},\ddot{\vec n}) \mapsto
\big(\Phi_\theta(\vec n),\, T_{\vec n}\Phi_\theta(\dot{\vec n}),\, T_{\vec n}\Phi_\theta(\ddot{\vec n})\big)$, which we refer to as the \emph{second--order tangent lift} of $\Phi_\theta$.

\begin{lemma}
\label{lem:axial_invariance}
If $V(\vec n)=\widetilde V(n_z)$, then $V$ is invariant under rotations about the
$z$--axis, that is, $V(\Phi_\theta(\vec n)) = V(\vec n),\, \forall\,\theta\in\R$. \end{lemma}

\begin{proof}
The action $\Phi_{\theta}$ preserves the $z$--component, i.e.\
$(\Phi_\theta(\vec n))\cdot e_z = \vec n\cdot e_z = n_z$. Since $V$ depends only
on $n_z$, the invariance follows immediately.
\end{proof}

Consequently, the second--order Lagrangian
\begin{equation}\label{eq:L_second_order}
   L(\vec n,\dot{\vec n},\ddot{\vec n})
   = \frac12\|D_t\dot{\vec n}\|^2 + V(\vec n)
\end{equation}
is invariant under the second--order tangent lift of $\Phi_\theta$:
the kinetic term is invariant because $\Phi_\theta$ is an isometry and the
potential term is invariant by Lemma~\ref{lem:axial_invariance}.

Let $\xi\in\mathfrak{so}(2)\cong\R$ denote the Lie algebra element generating the
one--parameter group $\Phi_\theta$. The associated infinitesimal generator is the
vector field $\displaystyle{X(\vec n) := \left.\frac{d}{d\theta}\right|_{\theta=0}\Phi_\theta(\vec n)
   = e_z\times \vec n}$, which is a \emph{Killing vector field} for the round metric (i.e.\ its flow acts by
isometries). 

\begin{definition}
\label{def:second_order_momentum}
Along a smooth curve $\vec n(t)\in S^2$, define the second--order momentum map for axial symmetry as
\begin{equation*}
   J_z(t)
   := \Big\langle \tfrac{d}{dt}\big(D_t\dot{\vec n}(t)\big),\, X(\vec n(t)) \Big\rangle
      - \Big\langle D_t\dot{\vec n}(t),\, D_t X(\vec n(t)) \Big\rangle,
\end{equation*}
where $\langle\cdot,\cdot\rangle$ denotes the round metric on $S^2$ and
$D_tX(\vec n)=\nabla_{\dot{\vec n}}X$.
\end{definition}

\begin{theorem}
\label{thm:quasiMomentum}
Let $\vec n(t)$ be a solution of
\begin{equation}
   D_t^3\dot{\vec n}
   + R(D_t\dot{\vec n},\dot{\vec n})\dot{\vec n}
   + \nabla V(\vec n) = 0,
\end{equation}
and assume that $V(\vec n)=\widetilde V(n_z)$. Then the second--order momentum map for axial symmetry
$J_z(t)$ is conserved.
\end{theorem}

\begin{proof}
Since $V(\vec n)=\widetilde V(n_z)$ depends only on the $z$--component, it is
invariant under rotations about the $z$--axis. The corresponding infinitesimal
generator is $X(\vec n)=e_z\times \vec n$. Define the axial momentum
\[
   J_z
   :=
   \left\langle
      \frac{\partial L}{\partial \dot{\vec n}}
      -
      \frac{d}{dt}\frac{\partial L}{\partial \ddot{\vec n}},
      \,
      X(\vec n)
   \right\rangle
   -
   \left\langle
      \frac{\partial L}{\partial \ddot{\vec n}},
      \,
      \dot X(\vec n)
   \right\rangle ,
\]
where $\dot X(\vec n)=e_z\times\dot{\vec n}$.

Differentiating $J_z$ with respect to time,
\begin{align*}
\frac{d}{dt}J_z
&=
\left\langle
   \frac{d}{dt}\frac{\partial L}{\partial \dot{\vec n}}
   -
   \frac{d^2}{dt^2}\frac{\partial L}{\partial \ddot{\vec n}},
   \,
   X(\vec n)
\right\rangle
+
\left\langle
   \frac{\partial L}{\partial \dot{\vec n}}
   -
   \frac{d}{dt}\frac{\partial L}{\partial \ddot{\vec n}},
   \,
   \dot X(\vec n)
\right\rangle
\\
&\quad
-
\left\langle
   \frac{d}{dt}\frac{\partial L}{\partial \ddot{\vec n}},
   \,
   \dot X(\vec n)
\right\rangle
-
\left\langle
   \frac{\partial L}{\partial \ddot{\vec n}},
   \,
   \ddot X(\vec n)
\right\rangle .
\end{align*}

The middle terms cancel identically. Using $\ddot X(\vec n)=e_z\times\ddot{\vec n}$,
we obtain
\[
\frac{d}{dt}J_z
=
\left\langle
   \frac{d}{dt}\frac{\partial L}{\partial \dot{\vec n}}
   -
   \frac{d^2}{dt^2}\frac{\partial L}{\partial \ddot{\vec n}},
   \,
   e_z\times\vec n
\right\rangle
-
\left\langle
   \frac{\partial L}{\partial \ddot{\vec n}},
   \,
   e_z\times\ddot{\vec n}
\right\rangle .
\]

Using the second-order Euler--Lagrange on $S^2$,
\[
\frac{d}{dt}\frac{\partial L}{\partial \dot{\vec n}}
-
\frac{d^2}{dt^2}\frac{\partial L}{\partial \ddot{\vec n}}
=
\frac{\partial L}{\partial \vec n}
+
\lambda\,\vec n ,
\]
where $\lambda$ enforces the constraint $\|\vec n\|=1$. Since $(e_z\times\vec n)\perp \vec n$, the constraint term drops out. Moreover,
\[
\left\langle
   \frac{\partial L}{\partial \vec n},
   \,
   e_z\times\vec n
\right\rangle
=
\left\langle
   \nabla V(\vec n),
   \,
   e_z\times\vec n
\right\rangle
=0
\]
because $V$ depends only on $n_z$. Finally, $\displaystyle{\left\langle
   \frac{\partial L}{\partial \ddot{\vec n}},
   \,
   e_z\times\ddot{\vec n}
\right\rangle
=0}$ since it is the inner product of a vector with its own cross--product. Therefore, $\frac{d}{dt}J_z=0$, and $J_z$ is conserved along solutions of the Euler--Lagrange equations.
\end{proof}



In the axially symmetric case considered here, the potential
$V(\vec n)=\widetilde V(n_z)$ breaks the $SO(3)$ symmetry of the sphere while
preserving an $SO(2)$ invariance associated with rotations about the
$z$--axis. This situation fits precisely into the paradigm of \emph{reduction by
symmetry breaking} for higher--order Lagrangian systems
\cite{goodmancolombo_reduction_broken,goodmancolombo_EP_broken}. 

While full
reduction is no longer possible, the presence of a remaining continuous symmetry
gives rise to a conserved quantity, here realised as a second--order momentum map
associated with the Killing field $X(\vec n)=e_z\times\vec n$. The axially symmetric potential considered here admits a particularly transparent interpretation within the extended–potential framework. Let $M=S^2$ and define $V_{\mathrm{ext}}(\vec n,\alpha) := \widetilde V(\vec n\cdot \alpha),
\,\alpha\in M$, with the natural action of $SO(3)$ on $M$ given by $\Psi_g(\alpha)=g\alpha$.
For the distinguished parameter value $\alpha_0=e_z$, one recovers the physical
potential $V(\vec n)=V_{\mathrm{ext}}(\vec n,\alpha_0)$. By construction, $V_{\mathrm{ext}}$ is invariant under the diagonal action
$(\vec n,\alpha)\mapsto(g\vec n,g\alpha)$, thereby restoring full rotational
invariance at the extended level. Fixing $\alpha=\alpha_0$ breaks the $SO(3)$
symmetry down to the axial subgroup $SO(2)$, which explains the persistence of a
single conserved quantity.


\subsection{Quantum Riemannian cubics with obstacle avoidance on the projective Hilbert space}

Let $\mathcal{H}$ be a complex Hilbert space of finite dimension $n+1$,
endowed with a Hermitian inner product
$\langle \cdot , \cdot \rangle : \mathcal{H}\times\mathcal{H}\to\mathbb{C}$,
linear in the second argument and conjugate–linear in the first. The associated norm is given by
$\|\psi\| := \sqrt{\langle \psi,\psi\rangle}$. The space of pure quantum states is identified with the projective Hilbert space $\mathbb{P}(\mathcal{H}):=
  (\mathcal{H}\setminus\{0\})/\mathbb{C}^{\ast}$, where $\mathbb{C}^{\ast}:=\mathbb{C}\setminus\{0\}$ denotes the multiplicative
group of nonzero complex scalars acting on $\mathcal{H}\setminus\{0\}$ by complex
rescaling, $(\lambda,\psi)\mapsto \lambda\psi$. Elements of $\mathbb{P}(\mathcal{H})$ are equivalence classes $[\psi]=\{\lambda\psi:\lambda\in\mathbb{C}^{\ast}\}$,which represent physical pure states.
One may equivalently normalise vectors in $\mathcal{H}$ and
then identify states that differ by a global phase, leading to the same
projective space.

The manifold
$\mathbb{P}(\mathcal{H})$ is naturally diffeomorphic to the complex
projective space $\mathbb{C}P^{n}$.  We implicitly fix a unit-norm representative and identify tangent vectors with horizontal lifts orthogonal to the fibre. The tangent space at a point $[\psi]\in\mathbb{P}(\mathcal{H})$ admits the
representation
\[
  T_{[\psi]}\mathbb{P}(\mathcal{H})
  \;\simeq\;
  \bigl\{
    \delta\psi\in\mathcal{H}
    \;\big|\;
    \langle \psi , \delta\psi \rangle = 0
  \bigr\},
\]
where the orthogonality condition removes the vertical direction associated
with the $\mathbb{C}^{\ast}$--action. This identification is independent of
the choice of representative $\psi$. 
$\mathbb{P}(\mathcal{H})$ carries a canonical
Riemannian metric, known as the \emph{Fubini--Study metric}, induced by the
Hermitian structure of $\mathcal{H}$. Given two tangent vectors
$\delta\psi_1,\delta\psi_2\in T_{[\psi]}\mathbb{P}(\mathcal{H})$, this metric is defined by
\[
  g_{FS}([\psi])(\delta\psi_1,\delta\psi_2)
  \;:=\;
  \operatorname{Re}\,\langle \delta\psi_1 , \delta\psi_2 \rangle,
\]
where $\psi$ is assumed to be normalized, $\langle \psi,\psi\rangle = 1$.
This definition is invariant under changes of representative and
yields a well-defined Riemannian structure on $\mathbb{P}(\mathcal{H})$. Equipped with the Fubini--Study metric, $\mathbb{P}(\mathcal{H})$ becomes a
compact K\"ahler manifold, endowed with a canonical Levi--Civita connection \cite{anandan1990}.

The special unitary group $G := \mathrm{SU}(n+1)$ acts transitively on
$\mathbb{P}(\mathcal{H})$ by $[\psi] \mapsto [U\psi]$, $U \in G$, and the isotropy subgroup of a reference state $[\psi_0]$ is isomorphic
to $K := S(\mathrm{U}(1)\times \mathrm{U}(n)) \simeq \mathrm{U}(n)
$. Hence, $\mathbb{P}(\mathcal{H})$ can be identified
with the homogeneous space $\mathbb{P}(\mathcal{H}) \simeq G/K$. Throughout, we regard curves on $\mathbb{P}(\mathcal{H})$ as projected
orbits of curves on $G$, and we formulate the dynamics on $G$ while
ensuring compatibility with the quotient structure.

Let $x(t) \in \mathbb{P}(\mathcal{H})\simeq G/K$ be a smooth curve on the projective
space. A \emph{quantum Riemannian cubic} on $\mathbb{P}(\mathcal{H})$ is defined as
a critical point of the second-order variational problem
\eqref{eq:RCfunctional}.
To obtain an explicit formulation, we lift the problem to $G=SU(n+1)$. Let $g(t)\in SU(n+1)$ be a smooth curve and define the associated projected curve
$x(t):=\pi(g(t))$, where $\pi:SU(n+1)\to G/K$, $\pi(g):=gK$, denotes the canonical projection. Equivalently, fixing a reference vector
$\psi_0\in\mathcal H$, one may write $x(t)=[g(t)\psi_0]$. Let $\xi(t):=g(t)^{-1}\dot g(t)\in\mathfrak{su}(n+1)$ denote the left-trivialised
velocity, and endow $\mathfrak{su}(n+1)$ with the $\mathrm{Ad}$-invariant inner
product $\langle A,B\rangle:=-\operatorname{tr}(AB)$. This choice induces the
normal homogeneous metric on $G/K$, compatible with the Fubini--Study geometry.

In this lifted representation, the covariant acceleration of the
projected curve can be expressed in terms of $\xi$ and its derivatives.
Up to vertical components associated with the isotropy subgroup, the
Riemannian cubic problem on $\mathbb{P}(\mathcal{H})$ is equivalent to a higher-order
variational problem on $SU(n+1)$ of the form $\mathcal{S}[g]
  =
  \frac12 \int_{0}^{T}
  \big\| \dot{\xi}(t) \big\|^2 \, \mathrm{d}t$, subject to suitable horizontality conditions ensuring that the dynamics
projects correctly to $\mathbb{P}(\mathcal{H})$.

In the Schrödinger picture, the Lie algebra element
$iH(t) := \xi(t)$ corresponds to a time-dependent Hamiltonian
$H(t) \in \mathfrak{u}(n+1)$. The curve $g(t)$ satisfies the \textit{Schrödinger
equation} $\dot g(t) = - i H(t) g(t)$, and the projected curve $x(t) = [g(t)\psi_0]$ represents the evolution
of a pure quantum state. The action functional measures the
total variation of the Hamiltonian through $||\dot H(t)||^{2}$, and quantum Riemannian cubics correspond to quantum evolutions that
minimise changes in the Hamiltonian while maintaining smoothness of the
state trajectory \cite{brody2012quantum}, \cite{abrunheiro2018general}.

To incorporate obstacle avoidance we introduce a
smooth potential $V : \mathbb{P}(\mathcal{H}) \longrightarrow \mathbb{R}$. The lifted potential $V_{\mathrm{ext}}(g) := V(\pi(g))$ is invariant under the action of the isotropy subgroup $K$, that
is, $V_{\mathrm{ext}}(gk) = V_{\mathrm{ext}}(g)$, $\forall\,k \in K$, and therefore defines a well-posed function on the homogeneous space.
Quantum Riemannian cubics with obstacle avoidance are given by
\begin{equation}
\label{eq:continuous-balance}
\frac{\mathrm d}{\mathrm dt}\mu
+
\operatorname{ad}^*_{\xi}\mu
=
-\,\operatorname{grad}_{G} V_{\mathrm{ext}}(g),
\,
\mu(t) := \dot{\xi}(t) \in \mathfrak{g}^*,
\end{equation}
where $\operatorname{grad}_{G} V_{\mathrm{ext}}(g)\in\mathfrak{g}^*$ denotes the
Riemannian gradient of the lifted potential with respect to the normal
homogeneous metric on $G$, identified via left trivialisation.
Owing to the $K$--invariance of $V_{\mathrm{ext}}$, this gradient is horizontal
and therefore projects consistently onto the base manifold
$\mathbb{P}(\mathcal{H})\simeq G/K$. Identifying $\xi(t)=iH(t)$ with a time-dependent Hamiltonian
$H(t)\in\mathfrak{u}(n+1)$, equation~\eqref{eq:continuous-balance} yields the
modified \emph{second-order Schr\"odinger equation} with potential
\begin{equation}
\label{eq:second-order-schrodinger}
\ddot H \;+\; i[H,\dot H]
\;=\;
-\,\operatorname{grad}_{G} V_{\mathrm{ext}}(g),
\end{equation}
where the right-hand side is understood via left trivialisation and the trace
inner product $\langle A,B\rangle=-\operatorname{tr}(AB)$. Given a solution $\xi(t)=iH(t)$ of \eqref{eq:second-order-schrodinger}, the unitary
evolution is recovered from $\dot g(t)=-iH(t)g(t)$ with initial condition
$g(0)=I$. The corresponding quantum state trajectory is then obtained as
$x(t)=[g(t)\psi_0]\in\mathbb{P}(\mathcal{H})$.

\section{Variational Integrators for Quantum Riemannian Cubics with Obstacle Avoidance}

We briefly recall the framework of Lie group variational integrators
(LGVI) for second-order variational problems on Lie groups  \cite{colombo2012discrete}, \cite{colombo2015variational}.  Let $G$ be a Lie group, and let $h>0$ denote the time step. A discrete 
path is a sequence $(g_0,\dots,g_N)\subset G$ approximating $g(t)$ at
times $t_k = kh$. For a Lagrangian $L:TG\to\mathbb{R}$, a discrete
Lagrangian $L_d : G\times G\to\R$ approximates the action integral over
one time step: $L_d(g_k,g_{k+1}) \approx \int_{t_k}^{t_{k+1}} L(g,\dot g)\,dt$. The discrete action sum $\displaystyle{\mathcal{S}_d(g_0,\dots,g_N) = 
  \sum_{k=0}^{N-1} L_d(g_k,g_{k+1})}$ is then extremised over variations in $g_k$ that fix $g_0$ and $g_N$,
leading to discrete Euler--Lagrange equations on $G\times G$.

For a second-order Lagrangian $L(g,\dot g,\ddot g)$, it is natural to
use a discrete Lagrangian $L_d : G\times G\times G \to \R$ depending on
triples $(g_{k-1},g_k,g_{k+1})$ to approximate the action integral over a two-step time interval. A map $L_d : G\times G\times G \to \R$ is a second-order discrete Lagrangian if $\displaystyle{L_d(g_{k-1},g_k,g_{k+1})
  \approx \int_{t_{k-1}}^{t_{k+1}} L(g,\dot g,\ddot g)\, dt}$ for curves $g(t)$ satisfying $g(t_{k-1})=g_{k-1}$,
$g(t_k)=g_k$, $g(t_{k+1})=g_{k+1}$.

Given a discrete path $(g_0,\dots,g_N)$ we define the discrete action for $L_d$ as 
\begin{equation}
\label{eq:discreteAction}
  \mathcal{S}_d(g_0,\dots,g_N)
  =
  \sum_{k=1}^{N-1} L_d(g_{k-1},g_k,g_{k+1}).
\end{equation}

A sequence $(g_0,\dots,g_N)\subset G$ is a  is a critical point of the discrete
action \eqref{eq:discreteAction} with fixed $(g_0,g_1)$ and $(g_{N-1},g_N)$ if and only if for each $k=1,\dots,N-1$, the discrete second-order Euler--Lagrange
equations
\begin{align}
\label{eq:DEL2}
 0=&D_2 L_d(g_{k-1},g_k,g_{k+1})
  + D_3 L_d(g_{k-2},g_{k-1},g_k)\nonumber
  \\&+ D_1 L_d(g_k,g_{k+1},g_{k+2})
\end{align}
hold, where $D_i L_d$ denotes the differential of $L_d$ with respect to its
$i$-th argument, viewed as an element of the corresponding cotangent
space via left trivialisation (see \cite{colombo2012discrete}, \cite{colombo2015variational}). This choice of boundary conditions is the discrete analogue of fixing
position and velocity at the endpoints in the continuous-time
variational problem \cite{colombo2016geometric}.

\subsection{Variational integrators for Riemannian cubics with obstacle avoidance}

We now construct a second-order variational integrator for Riemannian
cubics with obstacle avoidance on a homogeneous space $G/K$. For
simplicity, the formulation is presented on the Lie group $G$; the
projection to the quotient space $G/K$ is then handled through the
homogeneous structure.

Let $h>0$ denote the time step and define the group increments $W_k := g_k^{-1} g_{k+1} \in G$. For sufficiently small $h$, each increment is represented by a reduced
Lie algebra element $\xi_k \in \frakg$ via a retraction map. Throughout
this work we employ the Cayley transform for matrix Lie groups, $W_k = \mathrm{cay}(h\,\xi_k)
  := (I - \tfrac{h}{2}\xi_k)^{-1}(I + \tfrac{h}{2}\xi_k)$, which provides a second-order accurate and symmetric approximation of
the exponential map \cite{colombo2012discrete}.

A second-order approximation of the covariant acceleration is obtained via the centred second
difference $\displaystyle{a_k :=
  \frac{\xi_{k+1} - 2\xi_k + \xi_{k-1}}{h^2}}$. Let $\langle\cdot,\cdot\rangle$ be an inner product on $\frakg$
corresponding to a left-invariant Riemannian metric on $G$. The discrete
kinetic term is defined as $K_d(\xi_{k-1},\xi_k,\xi_{k+1})
  =
  \frac{h}{2}\,\|a_k\|^2$. The obstacle potential is discretised using a midpoint quadrature rule.
Let $g_{k+\frac12} := g_k\,\mathrm{cay}\!\left(\tfrac{h}{2}\xi_k\right)$ denote the midpoint configuration along the discrete trajectory. The
discrete potential term is defined as $V_d(g_{k-1},g_k,g_{k+1})
  =
  h\,V_{\mathrm{ext}}(g_{k+\frac12})$, where $V_{\mathrm{ext}} : G \to \mathbb{R}$ is the pullback of an obstacle
potential $V : G/K \to \mathbb{R}$ under 
$\pi : G \to G/K$. 

The discrete second-order Lagrangian $L_d:G\times G\times G\to\mathbb{R}$ for the Riemannian cubics
with obstacle avoidance is defined as $L_d=K_d+V_d$. We assume that the discrete Lagrangian satisfies a
regularity condition, namely that the mixed second derivative with
respect to consecutive configurations is non-singular. Under this regularity hypothesis, the discrete
second-order Euler--Lagrange equations define a \emph{local discrete
flow} $\Upsilon_{L_d} : G \times G \longrightarrow G \times G$, $\Upsilon_{L_d}(g_{k-1}, g_k) = (g_k, g_{k+1})$, by the implicit function theorem, for sufficiently small time steps.


Introducing the reduced variables $\xi_k$ via the Cayley transform and
defining $\mu_k := a_k \in \frakg^*$, identified with $\frakg$ through the chosen inner product, the discrete
second-order Euler--Lagrange equations are 
\begin{equation*}
\operatorname{Ad}_{W_k^{-1}}^* \mu_k
-
\operatorname{Ad}_{W_{k-1}^{-1}}^* \mu_{k-1}
=
\mu_{k+1} - \mu_k
-
h\,\mathbf{d}V_{\mathrm{ext}}(g_{k+\frac12}),
\end{equation*}
where $\operatorname{Ad}^*$ denotes the coadjoint action and
$\mathbf{d}V_{\mathrm{ext}}(g) \in \frakg^*$ is the left-trivialised
gradient of the obstacle potential.

Once the reduced variables have been determined, the configuration
update on the group is obtained from the reconstruction equation $g_{k+1} = g_k W_k$. The corresponding physical trajectory on the homogeneous space is then
given by the projected sequence $x_k := \pi(g_k) \in G/K$. Owing to the $K$-invariance of the discrete Lagrangian and of the
obstacle potential, the sequence $(x_k)$ is independent of the chosen
representative $g_k$ and defines a well-posed discrete trajectory on
$G/K$.

\subsection{Variational integrators for quantum Riemannian cubics with obstacle avoidance}
\label{subsec:VI_projective}

To construct a structure--preserving numerical scheme, we discretise
the variational principle directly on $\mathbb P(\mathcal H)$.
Let $h>0$ be a fixed time step and let $\big([\psi_0],[\psi_1],\dots,[\psi_N]\big)
   \subset \mathbb P(\mathcal H)$ be a discrete path approximating a trajectory $\gamma$ at times $t_k=kh$. The discrete Lagrangian $L_d:\mathbb P(\mathcal H)\times\mathbb P(\mathcal H)\times\mathbb P(\mathcal H) \to \mathbb R$ is chosen as an approximation of the integral of the continuous Lagrangian over one time step.
Specifically, we define
\begin{equation}\label{eq:discrete_Lagrangian_projective}
   L_d([\psi_{k-1}],[\psi_k],[\psi_{k+1}])
   =
   \frac{1}{2h}\,
   \big\|\mathcal A_k\big\|^2
   + h\,V([\psi_k]),
\end{equation}
where $\|\cdot\|$ denotes the norm induced by the Fubini--Study metric on
$\mathbb P(\mathcal H)$, and
$\mathcal A_k\in T_{[\psi_k]}\mathbb P(\mathcal H)$ is a
discrete approximation of the covariant acceleration at time $t_k$, constructed
from $([\psi_{k-1}],[\psi_k],[\psi_{k+1}])$.
It is defined intrinsically using the Riemannian logarithm associated
with the metric. Let $\psi_k\in\mathcal H$ be a unit-norm
representative of $[\psi_k]$. For $j\in\{k-1,k+1\}$, define
\[
  \log_{[\psi_k]}([\psi_j])
  :=
  \frac{\theta_{kj}}{\sin\theta_{kj}}
  \bigl(
    \psi_j - \langle\psi_k,\psi_j\rangle\,\psi_k
  \bigr)
  \;\in\;
  T_{[\psi_k]}\mathbb P(\mathcal H),
\]
where
\(
  \theta_{kj} := \arccos\bigl(|\langle\psi_k,\psi_j\rangle|\bigr)
\)
is the distance between $[\psi_k]$ and $[\psi_j]$.
The discrete acceleration is then given by
\[
  \mathcal A_k
  :=
  \frac{1}{h^2}
  \left(
    \log_{[\psi_k]}([\psi_{k+1}])
    -
    \log_{[\psi_k]}([\psi_{k-1}])
  \right).
\]

The discrete action sum is then given by\newline
$\displaystyle{\mathcal S_d
   =\sum_{k=1}^{N-1}
   L_d([\psi_{k-1}],[\psi_k],[\psi_{k+1}])}$, with fixed boundary points
$[\psi_0]$, $[\psi_1]$, $[\psi_{N-1}]$, and $[\psi_N]$. Critical points of the discrete action with respect to variations of the interior points
$([\psi_1],\dots,[\psi_{N-1}])$
yield the discrete second-order Euler--Lagrange equations \eqref{eq:DEL2} for this second--order
variational problem which defines an implicit update rule
relating five consecutive points of the discrete trajectory, $([\psi_{k-2}],[\psi_{k-1}],[\psi_k],[\psi_{k+1}],[\psi_{k+2}])$. Under the regularity assumption on $L_d$, this implicit relation determines a discrete local flow on $\mathbb P(\mathcal H)$.

For analytical purposes and, in particular, for efficient numerical
implementation, it is often convenient to lift this formulation to a
unitary group acting transitively on $\mathbb P(\mathcal H)$. Let $\mathcal H\simeq\mathbb C^{n+1}$ and consider the natural left action
of the unitary group $U(n+1)$ on $\mathcal H$.
Fix a reference state $\psi_0\in\mathcal H$ with $\|\psi_0\|=1$ and denote
by $[\psi_0]\in\mathbb P(\mathcal H)$ its associated ray.
The isotropy subgroup of this action is
\[
   K
   :=
   \big\{ U\in U(n+1)\ :\ U\psi_0 = e^{i\theta}\psi_0 \big\}
   \simeq U(1)\times U(n),
\]
so that the projective space may be identified with the homogeneous
space $\mathbb P(\mathcal H) \simeq U(n+1)/K$. Any discrete projective trajectory
$([\psi_0],\dots,[\psi_N])$
admits a lifted representation $[\psi_k] = \pi(U_k),
   \,
   U_k\in U(n+1)$, where the projection $\pi:U(n+1)\to\mathbb P(\mathcal H)$ is\ $\pi(U):=[U\psi_0]$.


Given a discrete Lagrangian as in~\eqref{eq:discrete_Lagrangian_projective},
we define its lifted counterpart $L_d^{\mathrm{ext}}:
   U(n+1)\times U(n+1)\times U(n+1) \to\mathbb R$ by
\begin{equation*}
   L_d^{\mathrm{ext}}(U_{k-1},U_k,U_{k+1})
   :=
   L_d\big(
      \pi(U_{k-1}),\,
      \pi(U_k),\,
      \pi(U_{k+1})
   \big).
\end{equation*} By construction, $L_d^{\mathrm{ext}}$ is invariant under the right action of the
isotropy subgroup $K$, namely, for $G\in K$
\[
   L_d^{\mathrm{ext}}(U_{k-1}G,U_kG,U_{k+1}G)
   =
   L_d^{\mathrm{ext}}(U_{k-1},U_k,U_{k+1}),
\]
and therefore descends to a discrete Lagrangian on $\mathbb P(\mathcal H)$.

Let $U_k\in U(n+1)$ be a lifted discrete trajectory and define the
left--trivialised discrete increments $W_k := U_k^{-1}U_{k+1}\in U(n+1)$. For sufficiently small time step $h>0$, each increment is parametrised by a Lie
algebra element $\Xi_k\in\mathfrak{u}(n+1)$ via the Cayley map, $W_k = \mathrm{cay}(h\,\Xi_k)$. The elements $\Xi_k$ represent discrete generators of unitary evolution and
constitute a discrete analogue of the Schr\"odinger Hamiltonian generator. With this notation, the lifted discrete Lagrangian
$L_d^{\mathrm{ext}}(U_{k-1},U_k,U_{k+1})$
can be written in terms of the variables $(U_k,\Xi_{k-1},\Xi_k)$ as
\begin{equation}\label{eq:Ld_Xi}
   L_d^{\mathrm{ext}}
   =
   \frac{h}{2}\Big\|
      \frac{\Xi_k-\Xi_{k-1}}{h}
   \Big\|^2
   + h\,V\big(\pi(U_k)\big),
\end{equation}
where the norm is induced by the bi--invariant inner product
$\langle X,Y\rangle := -\operatorname{Re}\operatorname{tr}(X^\dagger Y)$
on $\mathfrak{u}(n+1)$.
The discrete acceleration $\mathcal A_k$ is thus approximated by
$\frac{\Xi_k-\Xi_{k-1}}{h}$, which is invariant under the adjoint action of
$U(n+1)$.

Variations of $U_k$ induce variations of $\Xi_{k-1}$ and $\Xi_k$ through the
Cayley map \cite{colombo2012discrete}. Introducing the discrete momenta $\mu_k := \frac{\Xi_k-\Xi_{k-1}}{h}
   \;\in\;
   \mathfrak{u}(n+1)^*$, the discrete second-order Euler--Lagrange equations are:
\begin{equation}\label{eq:DEL_Ad}
   \operatorname{Ad}_{W_k^{-1}}^* \mu_k
   -
   \operatorname{Ad}_{W_{k-1}^{-1}}^* \mu_{k-1}
   =
   \mu_{k+1}-\mu_k
   -
   h\,\mathbf F_k,
\end{equation}
where $\operatorname{Ad}^*$ denotes the coadjoint action of $U(n+1)$ and
$\mathbf F_k\in\mathfrak{u}(n+1)^*$ is the discrete force induced by the obstacle
avoidance potential, defined by
\[
   \langle \mathbf F_k,\eta\rangle
   =
   \mathrm dV\big(\pi(U_k)\big)
   \cdot
   \big(
      \mathrm d\pi(U_k)\cdot (\eta U_k)
   \big),
   \,\, \forall\,\eta\in\mathfrak{u}(n+1).
\]

Equation~\eqref{eq:DEL_Ad} defines an implicit second--order difference
equation on the Lie algebra $\mathfrak{u}(n+1)$ coupled with the group
update $U_{k+1} = U_k\,\mathrm{cay}(h\,\Xi_k)$, which may be interpreted as a \emph{discrete second-order Schr\"odinger flow with
with obstacle avoidance}. By construction, $[\psi_k] = \pi(U_k)$ satisfies the  discrete second-order Euler--Lagrange equations on
$\mathbb P(\mathcal H)$.
Thus, solving the lifted system
\eqref{eq:DEL_Ad} yields a variational
integrator on the quantum state space after projection by $\pi$.

Let $\mathcal G\subset U(n+1)$ be a (closed) Lie subgroup acting by left
multiplication on $U(n+1)$. If the obstacle potential satisfies $V\big(\pi(GU)\big)=V\big(\pi(U)\big)$, $\forall\,G\in\mathcal G$, then the lifted discrete Lagrangian $L_d^{\mathrm{ext}}$ is invariant under this action.
The discrete Noether theorem for second-order variational problems implies the
conservation of the discrete momentum map
\begin{equation}\label{eq:discrete_momentum}
   J_k
   =
   \operatorname{Ad}^*_{W_{k-1}^{-1}}
   \Big(
      \frac{\Xi_k-\Xi_{k-1}}{h}
   \Big)
   \in \mathfrak g^*,
\end{equation}
where $\mathfrak g$ denotes the Lie algebra of $\mathcal G$.
Given a reduced discrete trajectory
$([\psi_k])\subset\mathbb P(\mathcal H)$
together with a fixed discrete momentum level
$J\in\mathfrak g^*$, a lifted trajectory on $U(n+1)$ may be reconstructed
by choosing generators $\Xi_k$ satisfying \eqref{eq:discrete_momentum} and updating $U_{k+1} = U_k\,\mathrm{cay}(h\,\Xi_k)$. 
\subsection{Qubit case: variational integrators on $SU(2)$}
\label{sec:qubit_LGVI}

For a single qubit, the relevant Lie group is $SU(2)$ with Lie algebra
$\mathfrak{su}(2)$.
Pure quantum states are represented by points on the Bloch sphere
$S^2\simeq\mathbb P(\mathbb C^2)$. A unitary $U\in SU(2)$ acts on a reference state $\ket{\psi_0}$ and
induces the Bloch vector $\vec n = \pi(U)\in S^2$, where $\pi:SU(2)\to S^2$ is the homogeneous projection associated with
the adjoint action.
An obstacle avoidance potential $V:S^2\to\mathbb R$ is lifted to $SU(2)$ by $V_{\mathrm{ext}}(U)=V(\pi(U))$.

The Lie algebra $\mathfrak{su}(2)$ consists of traceless skew--Hermitian
$2\times2$ matrices.
Using the Pauli basis
$\vec\sigma=(\sigma_x,\sigma_y,\sigma_z)$, any element
$X\in\mathfrak{su}(2)$ can be written as $X = -\frac{i}{2}\,\vec\omega\cdot\vec\sigma$,  $\vec\omega\in\mathbb R^3$. With the inner product $\langle X,Y\rangle=-2\,\mathrm{tr}(XY)$, this identification induces the Euclidean inner product on
$\mathbb R^3$: $\langle X,Y\rangle
   =
   \vec\omega\cdot\vec\eta$. Under this identification, the adjoint and coadjoint actions coincide
and are given by the vector cross product  $\operatorname{ad}^*_{\vec\omega}\vec v
   =
   \vec\omega\times\vec v$.

Let $(U_0,\dots,U_N)\subset SU(2)$ be a discrete unitary trajectory and
define the group increments $W_k = U_k^{-1}U_{k+1}$. Using the Cayley transform, we parametrise $W_k = \mathrm{cay}(hX_k)$, $X_k=-\frac{i}{2}\vec\omega_k\cdot\vec\sigma$. The group update rule $U_{k+1}
   =
   U_k\,\mathrm{cay}\!\left(
      -\frac{i h}{2}\vec\omega_k\cdot\vec\sigma
   \right)$ defines a discrete unitary Schr\"odinger flow.  The induced evolution on the Bloch sphere is obtained by projection.
Denoting $\vec n_k:=\pi(U_k)\in S^2$, the Cayley update induces the rotation $\vec n_{k+1}
   =
   R(\vec\omega_k)\,\vec n_k$, where $R(\vec\omega_k)\in SO(3)$ is the rotation associated with the adjoint
action of $\mathrm{cay}(hX_k)$. With this notation, the discrete second--order Lagrangian reads \[
   L_d(U_{k-1},U_k,U_{k+1})
   =
   \frac{h}{2}
   \Big\|
      \frac{\vec\omega_k-\vec\omega_{k-1}}{h}
   \Big\|^2
   +
   h\,V(\vec n_k).
\]

Introducing the discrete momenta $\vec\mu_k := \frac{\vec\omega_k-\vec\omega_{k-1}}{h}$, the discrete second--order Euler--Lagrange equations are
\[
   R_k\,\vec\mu_k
   -
   R_{k-1}\,\vec\mu_{k-1}
   =
   \vec\mu_{k+1}-\vec\mu_k
   -
   h\,\vec n_k\times\nabla V(\vec n_k),
\]
where $R_k=\operatorname{Ad}_{W_k^{-1}}^*\in SO(3)$ denotes the rotation induced
by the Cayley increment $W_k$. Using the identification $\operatorname{ad}^*_{\vec\omega}\vec v
=\vec\omega\times\vec v$, the above equation reduces, in the small--step limit,
to the vectorial form
\[
   \frac{\vec\omega_{k+1}-2\vec\omega_k+\vec\omega_{k-1}}{h}
   +
   \vec\omega_k\times
   \frac{\vec\omega_k-\vec\omega_{k-1}}{h}
   +
   \vec n_k\times\nabla V(\vec n_k)
   =
   0.
\]

Let $\hat R_\theta
   :=
   \exp\!\left(-\tfrac{i}{2}\theta\,\sigma_z\right)\in SU(2)$ denote axial rotations about the $z$--axis.
If $V(\vec n)=\widetilde V(n_z)$, then $L_d$ is invariant under the left action
$U_k\mapsto \hat R_\theta U_k$.
By the discrete Noether theorem, the axial component of the discrete momentum
is exactly preserved along the discrete flow, namely $J_z^{(k)}
   :=
   \vec\mu_k\cdot\vec e_z
   =
   \left(
      \frac{\vec\omega_k-\vec\omega_{k-1}}{h}
   \right)\cdot\vec e_z
   \equiv J_z$. Given a discrete Bloch trajectory $(\vec n_k)$ and a fixed value of $J_z$, the lifted unitary trajectory $(U_k)$ is reconstructed by
solving the discrete Euler--Lagrange equations subject to the constraint $\vec e_z\cdot
   \frac{\vec\omega_k-\vec\omega_{k-1}}{h}
   =
   J_z$, and updating $U_k$ via the Cayley flow
$U_{k+1}=U_k\,\mathrm{cay}\!\left(-\tfrac{i h}{2}\vec\omega_k\cdot\vec\sigma\right)$.

We now illustrate the above variational integrator by means of numerical simulations. We consider a single qubit with Bloch vector $\vec n(t)\in S^2$ and construct a quantum Riemannian cubic with obstacle avoidance connecting two prescribed boundary states while remaining outside a prescribed \emph{forbidden region} of the Bloch sphere. Such regions naturally arise in quantum control applications, for instance to limit population of decoherence–sensitive states or to avoid leakage subspaces.

As a representative test case, we choose $\vec n(0)=(0,0,-1)$ (the south pole of the Bloch sphere) and impose a terminal condition $\vec n(T)$ located just outside a forbidden spherical cap centred at the north pole. The initial velocity is tangent to the sphere and oriented along a simple great–circle direction, while the terminal velocity is left free.
The forbidden set is modelled as a neighbourhood of the north pole $\vec n_\star=(0,0,1)$, representing a region of the Bloch sphere associated with undesirable or fragile quantum states. To encode this constraint we introduce an axially symmetric obstacle potential depending only on the $z$–component of the Bloch vector,
\begin{equation}
   V(\vec n)
   =
   \widetilde V(n_z)
   =
   \frac{\tau}{1 + \Big(\dfrac{1 - n_z}{D}\Big)^{2N}},
\end{equation}
with parameters $\tau>0$, $D>0$, and $N\in\mathbb N$.
This penalises proximity to an entire spherical cap around $n_z=1$, and therefore models the avoidance of a forbidden region with finite angular thickness. The level set $d_{S^2}(\vec n,P)=D$ defines the boundary of the forbidden region, while the potential grows inside it as $N$ increases. 

The time interval $[0,T]$ is discretised into $N$ uniform steps of size $h=T/K$, and we consider a discrete lifted trajectory $(U_0,\dots,U_N)\subset SU(2)$ satisfying $\pi(U_0)=\vec n(0)$ and $\pi(U_N)=\vec n(T)$, where $\pi:SU(2)\to S^2$ denotes the  projection onto the Bloch sphere. The resulting discrete second–order Euler–Lagrange equations define a nonlinear boundary–value problem for the sequence ($U_k$), which is solved globally as a critical point of the discrete action subject to fixed boundary conditions, rather than by forward time marching. This implicit variational formulation is particularly well suited to receding–horizon and MPC–type implementations, where feasibility with respect to state constraints is essential. In practice, an equivalent shooting formulation on the initial generator may also be employed without altering the variational structure.

The discrete variational solver for quantum Riemannian cubics with obstacle
avoidance is summarized as follows:

\begin{table}[h!]
\label{tab:solver}
\begin{tabular}{p{2cm} p{5.8cm}}

\textbf{Stage} & \textbf{Operation} \\

Initialisation
&
Set horizon $T$ and steps $K$, define $h=T/K$, and initialise interior unitaries
$(U_1,\dots,U_{K-1})$ by geodesic interpolation or warm start. \\[0.3em]

Increment computation
&
Compute Lie algebra increments $X_k\in\mathfrak{su}(2)$ from
$U_{k+1}=U_k\,\mathrm{cay}(hX_k)$. \\[0.3em]

Variational residual
&
Form the second--order discrete Euler--Lagrange residual
$R_k(U_{k-2},\dots,U_{k+2})$ for $k=2,\dots,K-2$. \\[0.3em]

Update step
&
Perform a Lie--group Newton or Gauss--Seidel step as part of a global nonlinear
solve of the discrete Euler--Lagrange equations,
$U_k\leftarrow\exp(\delta_k)U_k$.
\\[0.3em]

Iteration
&
Repeat residual evaluation and updates until convergence. \\[0.3em]

Projection
&
Recover Bloch vectors $\vec n_k=\pi(U_k)\in S^2$. \\
\end{tabular}
\end{table}

We consider the boundary–value problem described above, namely a transition from the south pole of the Bloch sphere to a terminal state located just outside the forbidden region surrounding the north pole. The obstacle potential parameters are chosen as $\tau=30$, $D=0.35$, and $N=6$, yielding a stiff but smooth repulsive barrier around the forbidden cap. The final time is set to $T=1$, and the interval is discretised using $N=100$ time steps, corresponding to a step size $h=0.01$.

The parameter $D$ controls the angular extent of the forbidden region, while $N$ governs the sharpness of the repulsive barrier. In the qubit setting, the $z$–component of the Bloch vector satisfies $n_z = 2|\langle 0|\psi\rangle|^2 - 1$, so that values $n_z\approx 1$ correspond to high population of the computational basis state $|0\rangle$. $D=0.35$ therefore penalises trajectories already when $|\langle 0|\psi\rangle|^2 \gtrsim 0.82$, modelling avoidance of a high–population region. In the numerical experiment, the terminal condition $\vec n(T)$ is prescribed just outside the forbidden region, at a geodesic distance $d_{S^2}(\vec n(T),P)=D+\varepsilon$ with $\varepsilon>0$ small. This ensures that the boundary condition itself is admissible, while still forcing the trajectory to evolve close to the boundary of the forbidden cap as we see in Fig. \ref{fig:qubit_obstacle}.

\begin{figure}[h!]
\centering
\includegraphics[width=0.9\linewidth]{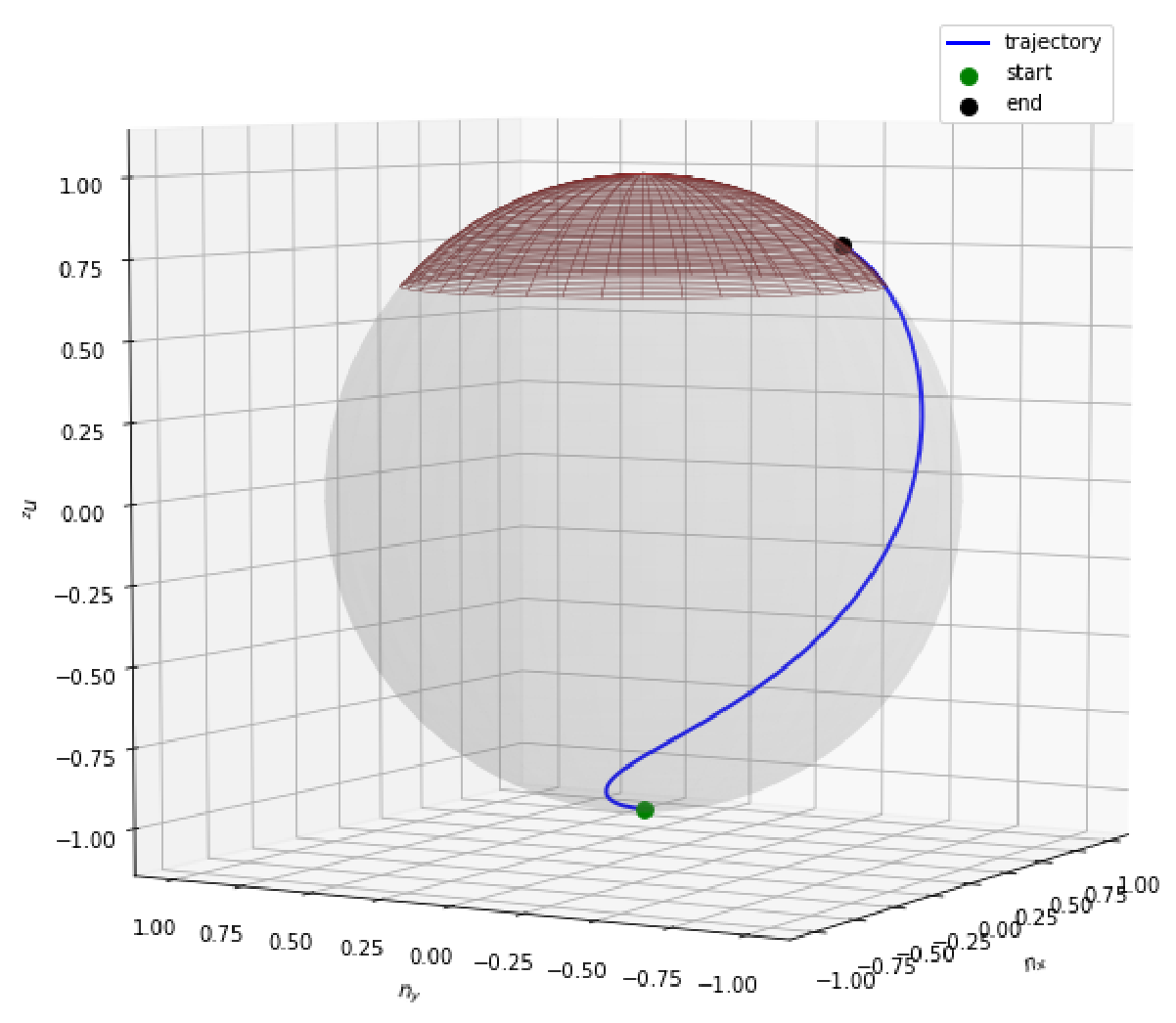}
\caption{Quantum Riemannian cubic  connecting the south pole to a terminal state located just outside a forbidden spherical cap around the north pole. The red cap denotes the forbidden region induced by the obstacle potential, while the blue curve shows the resulting cubic-shape trajectory.}
\label{fig:qubit_obstacle}
\end{figure}

To assess the geometric properties of the LGVI, we focus on the preservation of
the configuration constraint and discrete momentum conservation.
 As shown in the Fig \ref{fig:structure_preservation}, explicit Rungue-Kutta (RK) schemes violate the
spherical constraint $|\vec n|=1$ unless augmented with an ad hoc
projection or correction step, while the LGVI preserves the
constraint as a consequence of its variational
construction. Moreover, due to the symmetry of the discrete Lagrangian, the LGVI
preserves the associated discrete momentum map, in agreement
with the discrete Noether theorem. Since standard RK schemes do
not arise from a discrete variational principle, no corresponding
discrete momentum map exists, and no exact discrete conservation law can
be expected for such methods.

\begin{figure}[t]
  \centering
  \begin{minipage}{0.45\textwidth}
    \centering
    \includegraphics[width=\textwidth]{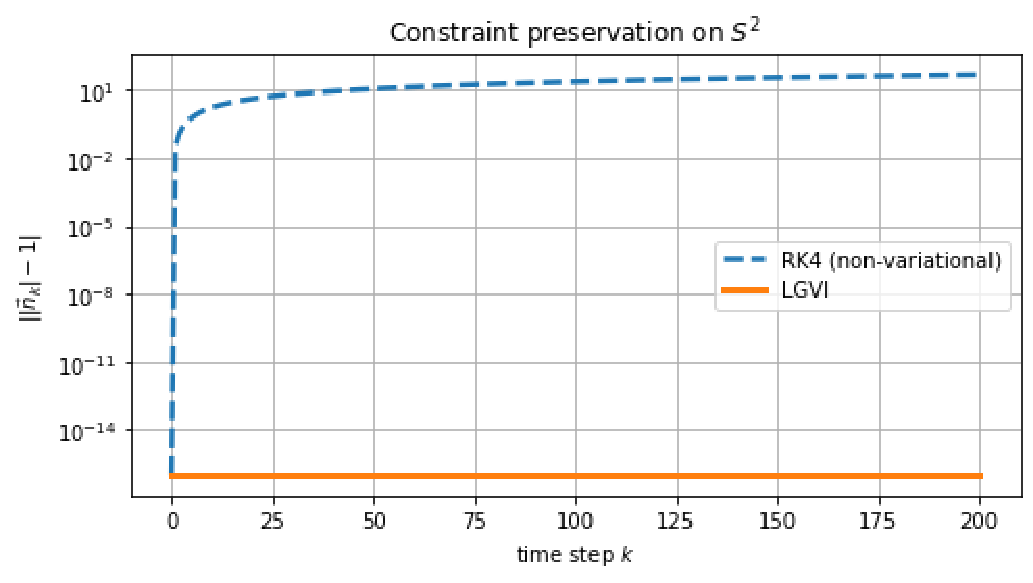}
  \end{minipage}
  \hfill
  \begin{minipage}{0.45\textwidth}
    \centering
    \includegraphics[width=\textwidth]{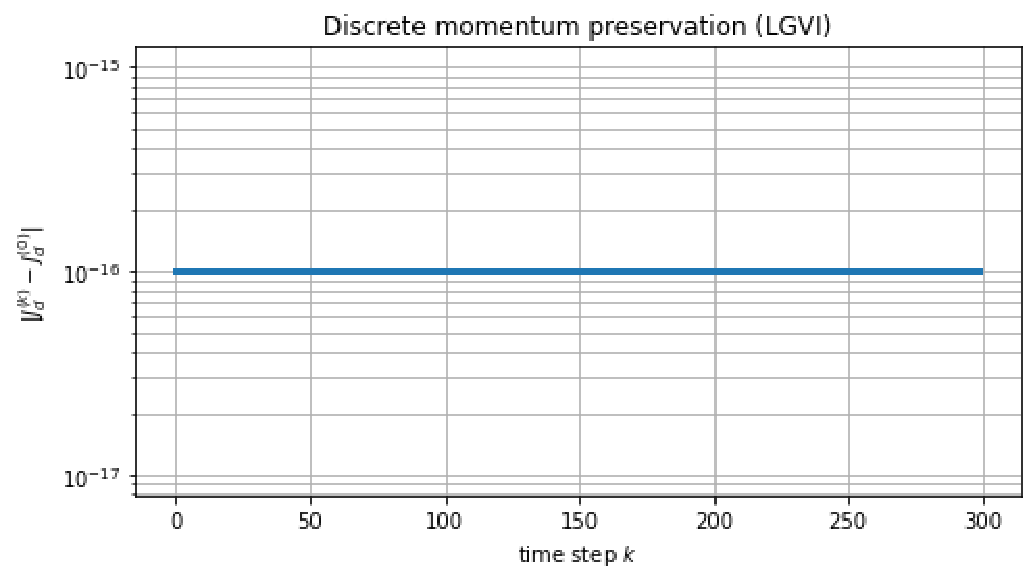}
  \end{minipage}

  \caption{\emph{Top:} Constraint preservation on $S^2$.
  The LGVI preserves the constraint intrinsically, whereas extrinsic
  explicit schemes exhibit constraint drift.
  \emph{Bottom:} Exact preservation of the discrete momentum map by the
  LGVI.
  }
  \label{fig:structure_preservation}
\end{figure}

\section{Quantum geometric model predictive control}

Model predictive control (MPC) provides a systematic framework for the
closed-loop control of dynamical systems subject to constraints, by
repeatedly solving finite-horizon optimal control problems and applying
the resulting control in a receding-horizon fashion. In the quantum
setting, however, the direct application of MPC faces intrinsic
difficulties arising from the geometry of the state space, the presence
of unobservable degrees of freedom, and the need for smooth control
signals compatible with physical implementations. 

Next, we introduce a \emph{quantum geometric model predictive
control} (QGMPC) framework, in which prediction, optimization, and
feedback are formulated intrinsically on the reduced quantum state
manifold. Pure quantum states are represented as points on the
projective Hilbert space $\mathcal M=\mathbb P(\mathcal H)$, endowed with
the Fubini--Study metric, and the global phase symmetry is eliminated
\emph{a priori} by symmetry reduction. This viewpoint removes
redundant degrees of freedom and yields a well-conditioned configuration
space for receding-horizon optimization. Unlike standard quantum control formulations based on first-order
dynamics, the proposed QGMPC framework is built upon a second-order
geometric model on $\mathcal M$, in which control inputs act at the level
of covariant acceleration. This choice is motivated by the need to
generate intrinsically smooth state trajectories, to penalize rapid
variations of the quantum state, and to ensure numerical robustness of
the finite-horizon optimization problem.

Let $\psi(t)\in\mathbb P(\mathcal H)$ denote the pure quantum state and
$v(t)\in T_{\psi(t)}\mathbb P(\mathcal H)$ its velocity. The dynamics are
described intrinsically by $\dot \psi(t) = v(t), \, \nabla^{FS}_{\dot \psi(t)} v(t) = u(t)$, where $\nabla^{FS}$ denotes the Levi--Civita connection associated with
the Fubini--Study metric and
$u(t)\in T_{\psi(t)}\mathbb P(\mathcal H)$ is a control input interpreted
as a covariant acceleration of the quantum state. The prediction problem underlying QGMPC is therefore
posed as the minimization of 
\begin{equation}\label{eq:qgmpc_continuous_cost}
\int_0^{T_p}
\left(
\frac12 \bigl\|\nabla^{FS}_{\dot \psi}\dot \psi\bigr\|_{FS}^2
+ V(\psi)
\right)\,dt,
\end{equation}
subject to the second-order dynamics and prescribed initial quantum
state and velocity. This formulation balances trajectory smoothness and
constraint avoidance in a manner that is intrinsic to the geometry of
$\mathbb P(\mathcal H)$ and fully consistent with quantum symmetry
reduction. Although posed as a continuous-time variational problem, the finite-horizon
action defines a predictive
model in the sense of model predictive control: for any measured or
estimated initial condition $(\psi_0,v_0)\in T\mathbb P(\mathcal H)$ and
prediction horizon $T_p>0$, it induces an optimal state trajectory
$\psi^\star(\cdot)$ on $[0,T_p]$, which can be repeatedly recomputed and
implemented in a receding-horizon fashion.

For feedback implementation, the problem is embedded in a discrete-time
model predictive control framework. Let $h>0$ denote the sampling period
and let $N_p\in\mathbb N$ define the prediction horizon $T_p=N_p h$.
Given the measured or estimated state $\psi_k\in\mathbb P(\mathcal H)$ at
time step $k$, we consider predicted state sequences $\{\psi_{k+j|k}\}_{j=0}^{N_p} \in (\mathbb P(\mathcal H))^{N_p+1}$, with initial condition $\psi_{k|k}=\psi_k$.

A discrete Lagrangian
$L_d:\mathbb P(\mathcal H)\times \mathbb P(\mathcal H)\times \mathbb P(\mathcal H)\to\mathbb R$ 
is introduced to approximate \eqref{eq:qgmpc_continuous_cost}. This
yields the finite-horizon discrete-time optimal control problem
\begin{align}
\min_{\{\psi_{k+j|k}\}} \;
J_{N_p}(\psi_k,v_k)
=
\sum_{j=2}^{N_p-2}
&L_d(\psi_{k+j-1:k+j+1|k})\nonumber
\\&+
\Phi(\psi_{k+N_p|k})\
\label{eq:quantum_mpc_problem} \\
\text{s.t.}\quad
 \psi_{k|k}=\psi_k,
\,\,
\psi_{k+j|k}&\in\mathbb P(\mathcal H),
\; j=0,\dots,N_p .\nonumber
\end{align} where $\psi_{a:b|k} := (\psi_{a|k},\dots,\psi_{b|k})$ and where the 
discrete velocities and accelerations are defined implicitly
through the chosen variational discretization on $\mathbb P(\mathcal H)$.
The terminal cost $\Phi:\mathbb P(\mathcal H)\to\mathbb R$ penalizes the
Fubini--Study distance to a desired target state, and no hard terminal
constraint is imposed. The necessary optimality conditions for
\eqref{eq:quantum_mpc_problem} are given by the discrete second-order Euler--Lagrange
equations for $L_d$, yielding a coupled nonlinear
system along the prediction horizon. At each sampling instant, this
system is solved using an iterative method initialized with the solution
from the previous horizon. The QGMPC action is defined by applying only the first
segment of the predicted trajectory, after which the horizon is
shifted forward and the procedure is repeated.

\begin{theorem}
\label{thm:quantum_mpc_stability}
Consider $\mathbb P(\mathcal H)$ endowed
with the Fubini--Study metric and the discrete-time closed-loop system
generated by the receding-horizon implementation of
\eqref{eq:quantum_mpc_problem}. Assume that:
(i) $L_d$ is continuous and positive definite
with respect to a desired equilibrium $\psi^\ast\in\mathbb P(\mathcal H)$; (ii) the potential $V$ attains a strict minimum at $\psi^\ast$ and is
positive definite in a neighborhood of $\psi^\ast$;
(iii) for a sufficiently large prediction horizon $N_p$, the finite-horizon
problem~\eqref{eq:quantum_mpc_problem} is feasible for all
$\psi_k$ in a neighborhood of $\psi^\ast$.
(iv) $\Phi$ is continuous, positive definite with
respect to $\psi^\ast$, and locally equivalent to $d_{FS}^2(\cdot,\psi^\ast)$.

Then the optimal value function
$J^\star_{N_p}(\psi_k)$ associated with
\eqref{eq:quantum_mpc_problem} is non-increasing along closed-loop
trajectories and $J^\star_{N_p}(\psi_{k+1})
\le
J^\star_{N_p}(\psi_k)
-
\ell_d(\psi_k)$, where $\ell_d$ denotes the discrete stage cost induced by $L_d$.
Consequently, the equilibrium $\psi^\ast$ is practically stable for the
closed-loop system, and the resulting state trajectory remains
bounded and converges to a neighborhood of $\psi^\ast$.
\end{theorem}

\begin{remark}
Theorem~\ref{thm:quantum_mpc_stability} shows that the optimal value
function of the finite-horizon variational problem acts as a Lyapunov
function for the closed-loop dynamics.
\end{remark}

\begin{proof}
Let $\psi^\ast\in\mathbb P(\mathcal H)$ denote the desired equilibrium state. By Assumption~(3), for a sufficiently large prediction
horizon $N_p$, the finite-horizon problem
\eqref{eq:quantum_mpc_problem} is feasible for all $\psi_k$ in a
neighborhood $\mathcal D\subset\mathbb P(\mathcal H)$ of $\psi^\ast$.
Since the receding-horizon implementation applies only the first segment
of an optimal predicted trajectory and shifts the horizon forward,
feasibility at time $k$ implies feasibility at time $k+1$ for all
$\psi_k\in\mathcal D$. Hence, the MPC  problem remains
well-posed along the closed-loop evolution. Define the optimal value function $J^\star_{N_p}(\psi_k)
=
\min_{\{\psi_{k+j|k}\}}
J_{N_p}(\psi_k)$, where $J_{N_p}$ is given by \eqref{eq:quantum_mpc_problem}. By
Assumption~(1), $L_d$ is continuous and positive
definite with respect to $\psi^\ast$. It follows that
$J^\star_{N_p}$ is continuous on $\mathcal D$ and positive definite with
respect to $\psi^\ast$.

Let $\{\psi^\star_{k+j|k}\}_{j=0}^{N_p}$ denote an optimal predicted
trajectory at time $k$, attaining the minimum
$J^\star_{N_p}(\psi_k)$. For $j=0,\dots,N_p-1$, consider the shifted sequence $\tilde\psi_{k+1+j|k+1} := \psi^\star_{k+1+j|k}$, which is feasible for the optimization problem at time $k+1$ by
construction and feasibility propagation. Evaluating the cost of this shifted sequence yields $J_{N_p}(\tilde\psi_{k+1})
=
J^\star_{N_p}(\psi_k)
-
\ell_d(\psi_k)
+
\Phi(\psi^\star_{k+N_p|k})
-
\Phi(\psi^\star_{k+N_p-1|k})$, where $\ell_d(\psi_k)$ denotes the discrete stage cost associated with
the first segment of the optimal predicted trajectory. By the positive
definiteness of the terminal cost $\Phi$ and its minimum at $\psi^\ast$,
the terminal contribution does not destroy the descent property. By optimality of
$J^\star_{N_p}(\psi_{k+1})$, we obtain $J^\star_{N_p}(\psi_{k+1})
\le
J^\star_{N_p}(\psi_k)
-
\ell_d(\psi_k)$.

By Assumptions~(1) and~(2), the stage cost $\ell_d$ is positive definite
with respect to $\psi^\ast$ in a neighborhood of $\psi^\ast$, and
$\ell_d(\psi_k)=0$ if and only if $\psi_k=\psi^\ast$. Therefore,
$J^\star_{N_p}$ is non-increasing along closed-loop trajectories and
strictly decreasing away from $\psi^\ast$. Since $J^\star_{N_p}$ is continuous, positive definite, and
non-increasing along the closed-loop evolution, standard discrete-time
Lyapunov arguments on manifolds imply that $\psi^\ast$ is a practically
stable equilibrium for the closed-loop system. In particular, for any
sufficiently small neighborhood $\mathcal U$ of $\psi^\ast$, there
exists a neighborhood $\mathcal V\subset\mathcal D$ such that any
closed-loop trajectory starting in $\mathcal V$ remains in
$\mathcal U$ for all future times and converges to a compact subset of
$\mathcal U$.
\end{proof}

We illustrate the proposed QGMPC framework with the case of a single qubit in Figure~\ref{fig:qgmpc_qubit}.
Starting from the south pole of the Bloch sphere, the quantum state is steered toward a target state located close to the boundary of a forbidden spherical cap centered at the north pole. $\Phi$ is implemented as a quadratic penalty on the
Fubini--Study distance to the target state, with terminal weight
$\alpha_T=25$.
The obstacle is encoded through a smooth repulsive potential of strength $\tau = 30$ and radius $D = 0.35$, while the prediction horizon is set to $N_p = 10$ with a sampling period $h = 0.05$. 
In the proposed QGMPC scheme, the predictive model at each sampling instant is defined by a LGVI, and therefore each finite-horizon prediction inherits the geometric structure of the underlying quantum dynamics.
The closed-loop evolution can be interpreted as a receding-horizon concatenation of local quantum Riemannian cubics with obstacle avoidance. 

\begin{figure}[h!]
  \centering
  \includegraphics[width=0.9\linewidth]{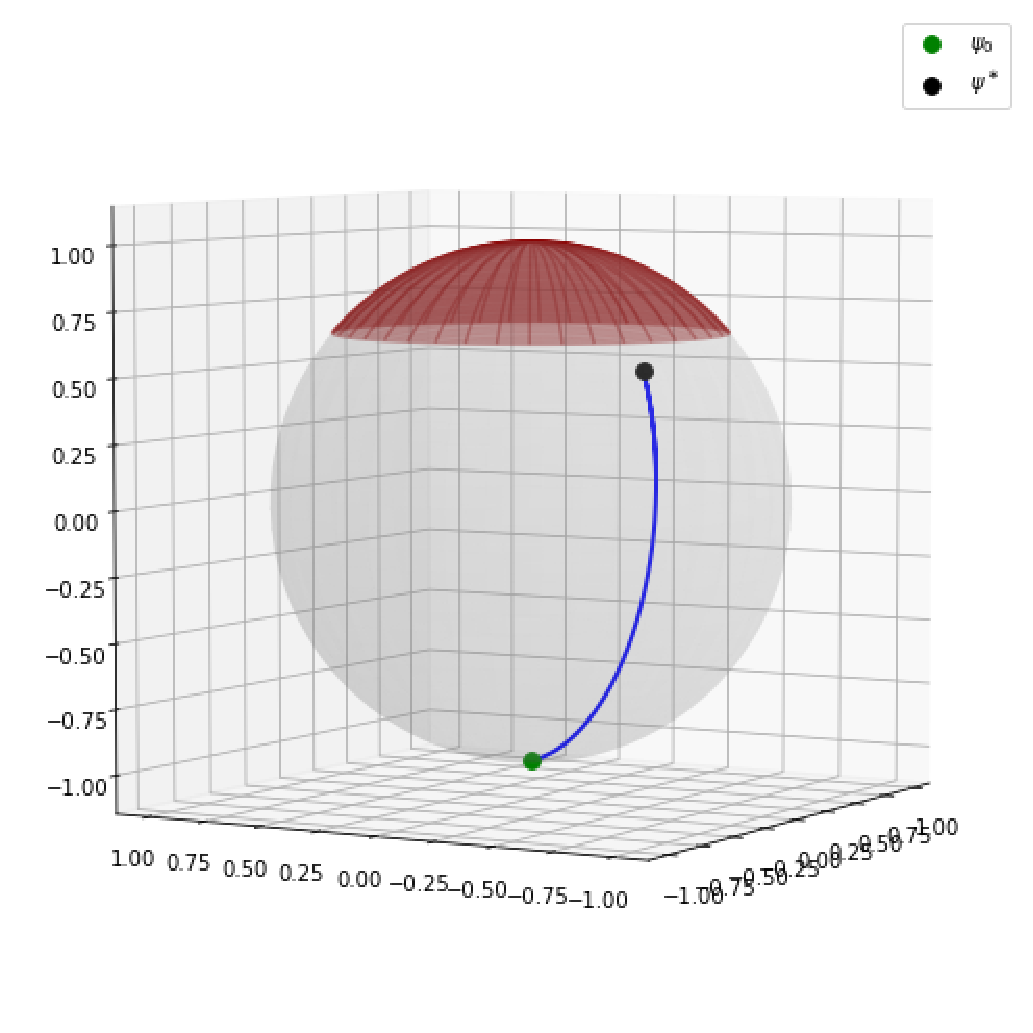}
  \caption{Closed-loop QGMPC for a single qubit.}
  \label{fig:qgmpc_qubit}
\end{figure}

To highlight the closed-loop nature of the proposed QGMPC framework, we compare its behavior with an open-loop strategy under a small state perturbation.
Both rely on the same geometric model, obstacle-avoidance potential, and LGVI-based discretization.
In the open-loop case, the trajectory is computed once and subsequently executed without correction, so deviations induced by the perturbation persist along the evolution.  In contrast, QGMPC recomputes a finite-horizon optimal prediction at each sampling instant using the LGVI, allowing the controller to compensate for the perturbation and steer the state back toward the desired target while respecting the forbidden region.

\begin{figure}[h!]
  \centering
  \includegraphics[width=0.9\linewidth]{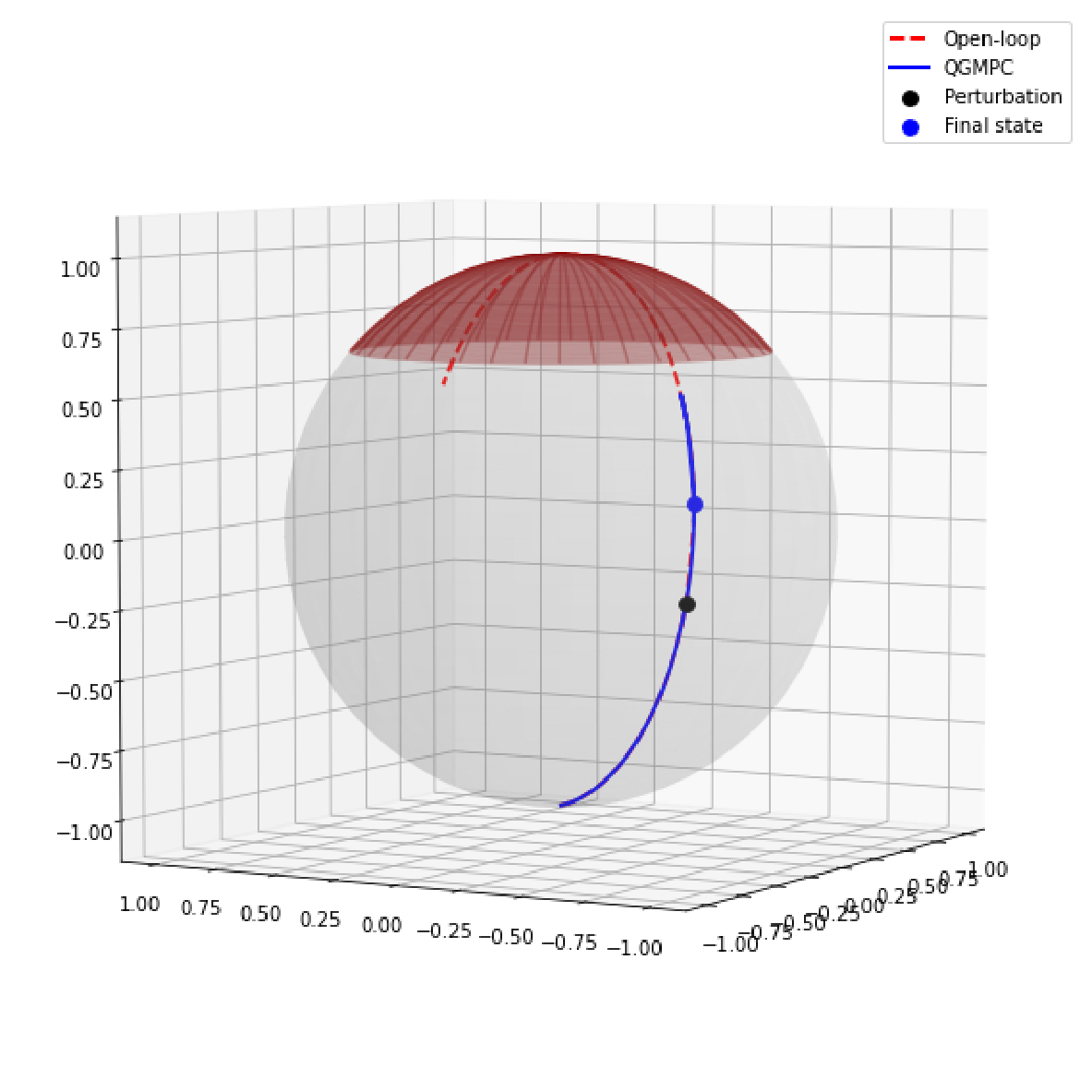}
  \caption{Closed-loop QGMPC vs. open loop}
  \label{fig:qgmpc_qubit2}
\end{figure}

We illustrate how the proposed QGMPC framework combines intrinsic geometric prediction with receding-horizon feedback to generate smooth, constraint-consistent trajectories in Fig. \ref{fig:qgmpc_qubit2}. The closed-loop QGMPC trajectory does not converge to the final state
exactly, nor to the boundary of the forbidden region. Instead, it
approaches a practically stable equilibrium determined by the balance
between terminal accuracy, obstacle avoidance, and trajectory
smoothness, in full agreement with the stability result of
Theorem~\ref{thm:quantum_mpc_stability}.

\section*{Conclusions}

We have presented a geometric framework for quantum trajectory generation
and control based on second-order variational principles on the projective
Hilbert space of pure quantum states. By penalizing covariant acceleration
and incorporating smooth obstacle-avoidance potentials, the approach
yields intrinsically smooth and constraint-aware quantum trajectories. The proposed formulation was embedded into a geometric model predictive
control scheme using structure-preserving variational discretizations.
A Lyapunov-type argument establishes practical closed-loop stability in the
absence of hard terminal constraints, and numerical simulations on the
Bloch sphere illustrate the robustness of the resulting QGMPC strategy in
the presence of disturbances.

\bibliographystyle{plain}        
\bibliography{autosam}           



\end{document}